\documentclass[letterpaper, 10 pt, conference]{ieeeconf}
\IEEEoverridecommandlockouts
\IEEEoverridecommandlockouts

\usepackage{stfloats}   

\usepackage{times}
\usepackage{mathptmx}
\usepackage{mathtools}
\usepackage{amssymb}
\mathtoolsset{showonlyrefs}
\usepackage{graphicx}
\usepackage{cite}
\usepackage{tikz}
\usetikzlibrary{arrows.meta, decorations.markings}
\definecolor{C10B981}{HTML}{10B981}
\definecolor{CECFDF5}{HTML}{ECFDF5}
\definecolor{C065F46}{HTML}{065F46}
\definecolor{C6B7280}{HTML}{6B7280}
\definecolor{C374151}{HTML}{374151}
\definecolor{C3B82F6}{HTML}{3B82F6}
\definecolor{CEFF6FF}{HTML}{EFF6FF}
\definecolor{C1E3A8A}{HTML}{1E3A8A}
\definecolor{C4B5563}{HTML}{4B5563}
\definecolor{C1D4ED8}{HTML}{1D4ED8}
\definecolor{C0EA5E9}{HTML}{0EA5E9}
\definecolor{CF0F9FF}{HTML}{F0F9FF}
\definecolor{C0C4A6E}{HTML}{0C4A6E}
\definecolor{C64748B}{HTML}{64748B}
\definecolor{C06B6D4}{HTML}{06B6D4}
\definecolor{CECFEFF}{HTML}{ECFEFF}
\definecolor{C164E63}{HTML}{164E63}
\definecolor{C0E7490}{HTML}{0E7490}
\definecolor{C075985}{HTML}{075985}
\definecolor{CA855F7}{HTML}{A855F7}
\definecolor{CFAF5FF}{HTML}{FAF5FF}
\definecolor{C581C87}{HTML}{581C87}
\definecolor{C7E22CE}{HTML}{7E22CE}
\definecolor{C260AF5}{HTML}{260AF5}
\definecolor{CEBEBFF}{HTML}{EBEBFF}
\definecolor{C78350F}{HTML}{78350F}
\definecolor{C160D91}{HTML}{160D91}
\definecolor{C1D0D91}{HTML}{1D0D91}
\definecolor{CF43F5E}{HTML}{F43F5E}
\definecolor{CFFF1F2}{HTML}{FFF1F2}
\definecolor{C9F1239}{HTML}{9F1239}
\definecolor{CBE123C}{HTML}{BE123C}
\definecolor{CF97316}{HTML}{F97316}
\definecolor{CFFFBEB}{HTML}{FFFBEB}
\definecolor{C7C2D12}{HTML}{7C2D12}
\definecolor{CC2410C}{HTML}{C2410C}
\definecolor{C0E0AF5}{HTML}{0E0AF5}

\usepackage{tikz}
\usetikzlibrary{arrows.meta}   

\usepackage{xcolor}

\usepackage{adjustbox}

\usetikzlibrary{arrows.meta}
\newtheorem{assumption}{\textbf{Assumption}}

\newtheorem{lemma}{\textbf{Lemma}}
\newtheorem{theorem}{\textbf{Theorem}}
\newtheorem{corollary}{\textbf{Corollary}}
\newtheorem{remark}{\textbf{Remark}}
\newtheorem{problem}{\textbf{Problem}}
\makeatletter
\@ifundefined{proof}{\newenvironment{proof}{\par\noindent{\it Proof: }}{\hfill$\blacksquare$\par}}{}
\makeatother
\DeclareMathOperator{\co}{co}
\DeclareMathOperator{\dist}{dist}

\title{\LARGE \bf
Resilient Output Containment under Undisclosed Leader Dynamics and Actuator Attacks
}
\author{Mohammadreza Nematollahi$^{1}$, Khashayar Khorasani$^{1}$, and Nader Meskin$^{2}$\thanks{$^{1}$Department of Electrical and Computer Engineering, Concordia University, Montreal, Canada. Emails: m\_nemato@encs.concordia.ca; kash@ece.concordia.ca.}%
\thanks{$^{2}$Department of Electrical Engineering, Qatar University, Doha, Qatar. Email: nader.meskin@qu.edu.qa.}%
\thanks{Supported by NATO (Emerging Security Challenges Division), NSERC Discovery, DND Supplemental, and DND IDEaS. The views expressed are those of the authors and not necessarily those of the sponsors or the Government of Canada.}%
}

\begin{document}
\maketitle
\thispagestyle{empty}
\pagestyle{empty}
\pagestyle{plain}        
\pagenumbering{arabic}   

\begin{abstract}
This work studies resilient output containment for heterogeneous linear multi-agent systems with actuator cyber-attacks over directed network topologies. The leaders generate bounded locally absolutely continuous trajectories; however, their dynamics, velocity bounds, and motion envelopes are undisclosed to the followers. The cyber-attack model includes state- and input-correlated, as well as bounded exogenous actuator false-data terms. A continuous two-layer adaptive control architecture is proposed. The first layer is a virtual-actuator reconfiguration layer that uses partial state measurements to compensate for actuator attacks in the local tracking-error dynamics. The second layer is a network interface that generates task-space commands via an adaptive interaction protocol. This protocol uses only neighbor-exchanged network-interface states whose dimensions match those of the plant output, and it does not require global graph knowledge for parameter tuning. For directed graphs, under a leader-rooted united spanning-tree condition, a nonsmooth Lyapunov analysis yields asymptotic containment at the command level. The physical outputs then converge to the leader convex hull up to a residual determined by the command-tracking local controllers. Simulation results using a network of quadrotors with damped suspended loads illustrate the performance of attack recovery and containment tracking.
\end{abstract}

\section{INTRODUCTION}
In multi-agent systems (MAS), containment tracking requires the outputs of a follower network to converge to the time-varying convex hull generated by multiple leaders \cite{cao2012distributed}. In security-critical cyber-physical deployments, two cybersecurity constraints complicate this objective. First, the leaders' dynamics often encode sensitive mission objectives or maneuvering patterns, and if exposed, this model knowledge equips adversaries with the information needed to synthesize stealthy cyber-attacks \cite{mustafa2020resilient}, necessitating strict nondisclosure policies that, in turn, restrict designers' access to such information. Second, such confidentiality does not, by itself, rule out local actuator-channel compromises, and adversaries can still corrupt the execution of containment commands. These disruptions may propagate through the network and undermine its collective objective. Therefore, local control architectures must compensate for admissible actuator effects while preventing their propagation into the coordination layer. 

These two considerations motivate a two-layer control design, namely a network layer that generates containment commands without relying on the models of leaders, paired with a local layer that guarantees execution resilience.

Distributed observer-based containment frameworks and their adaptive variants \cite{qin2018output, huang2022distributed, wang2022adaptive, zuo2018adaptive} typically introduce an estimation layer that requires followers to know or reconstruct, under suitable structural assumptions, the leaders' exosystem dynamic parameters, thereby conflicting with nondisclosure constraints. Sliding mode protocols \cite{lv2020adaptive, zhou2015containment, li2012distributed} can avoid this requirement by treating leader trajectories as bounded exogenous signals whose generator dynamics are unknown. However, these approaches have two main shortcomings: first, they often introduce discontinuous terms, leading to undesired chattering; and second, their reliance on \textit{a priori} knowledge of velocity bounds or motion envelopes results in conservative design. 

To alleviate chattering, the boundary-layer approximation of the discontinuous terms is used at the cost of demoting asymptotic convergence to mere ultimate boundedness. Continuous approximations with integrable residuals instead recover asymptotic convergence \cite{qu1994asymptotic, wu2004adaptive} and have recently been extended to multi-agent settings \cite{wang2023adaptive}. On the other hand, relaxing the requirement for known leaders' bounds requires adaptive mechanisms and has received comparatively little attention. Although the approach in \cite{wang2019tracking} addresses this issue using higher-order differentiators, the resulting protocol remains discontinuous and susceptible to chattering.

Continuous adaptive protocols, in turn, have been largely confined to undirected or detailed-balanced directed topologies \cite{xiao2020distributed, tian2022fully}, or are conditioned on local command approximability by adaptive PIDs \cite{lui2022leader}. Consequently, continuous asymptotic containment over general directed leader-rooted graphs, without leader-model reconstruction or known bounds on leaders' motion, remains underexplored.  

Observer-based containment designs separate network-level signal generation from local tracking through an explicit leader model or exosystem layer. This separation ensures that actuator attacks affect only the plant-level execution of a command, not the task-space protocol that generates it, so their effects remain local, non-propagating, and compensable. Under the undisclosed leader models considered here, this separation therefore cannot be inherited from a distributed observer layer and must instead be built directly into the architecture. 

Virtual-actuator and fault-hiding methods provide a natural mechanism for this purpose by inserting a recovery block between the nominal controller and the compromised actuator channel \cite{YADEGAR2021109514}. 

Existing Virtual-actuator results, however, typically provide ultimate boundedness under exogenous or fault-like attacks \cite{zuo2023resilient,jin2017adaptiveIII}, whereas asymptotic recovery results commonly require full state measurements \cite{xiao2020distributed}. Thus, continuous recovery with asymptotic tracking-error convergence under partial state measurements, while allowing combined state-correlated, input-correlated, and bounded exogenous actuator false-data terms, needs to be addressed.

To address the above gaps, this work develops a continuous two-layer architecture for resilient output containment of heterogeneous linear followers under actuator attacks and undisclosed leader dynamics. The contributions are as follows. First, a continuous adaptive protocol is proposed to generate local commands using only neighbor-exchanged network interface states, whose dimension matches the plant-output dimension. The protocol does not require global graph knowledge for tuning, nor does it require \textit{a priori} leader-velocity bounds, motion envelopes, or exosystem models. Second, a novel nonsmooth Lyapunov analysis is used to establish asymptotic command-containment results for the proposed protocol over directed graphs with a leader-rooted united spanning tree, without imposing symmetry restrictions on the follower subgraph. Third, a continuous adaptive virtual-actuator layer is proposed to asymptotically compensate for the actuator attacks, comprising state-correlated, input-correlated, and bounded false-data terms, while using only partial state measurements. Finally, the interconnection of the network-interface and virtual-actuator layers is shown to yield asymptotic convergence of the task-space command errors, while the physical outputs of the heterogeneous followers achieve practical containment.

\section{Problem Formulation}\label{sec:problem}
\subsection{Network Topology}
Consider a network of followers $\mathcal{F}$ and leaders $\mathcal{T}$, with $|\mathcal{F}|=M$ and $|\mathcal{T}|=N>1$, interacting over a fixed weighted digraph $\mathcal{G}(\mathcal{V},\mathcal{E})$, where $\mathcal{V}=\{1,\dots,M,M+1,\dots,M+N\}$ and $\mathcal{E}\subseteq \mathcal V\times \mathcal V$. Followers are indexed by \(i\in\mathcal F=\{1,\ldots,M\}\), while leaders are indexed by \(\ell\in\mathcal T=\{1,\ldots,N\}\). Let $\mathcal{A}_{\mathcal G}=[a_{ij}\geq 0]$ denote the adjacency matrix, with $a_{ij}>0$ whenever there is a directed link from agent $j$ to agent $i$, and let $\mathcal{L}_{\mathcal G}$ denote the graph Laplacian. Partitioning the agents into followers and leaders, after possibly permuting the agent indices, gives
\begin{equation}\label{eq: Laplacian_partition}
\mathcal{L}_{\mathcal G}=\begin{bmatrix}H_F & L_{FL}\\ \mathbf{0}_{N\times M} & \mathbf{0}_{N\times N}\end{bmatrix}.
\end{equation}
Here, $H_F\in\mathbb{R}^{M\times M}$ denotes the information-exchange matrix \cite{abdessameud2016distributed}, defined by $(H_F)_{ii}=\sum_{j\in\mathcal F}a_{ij}+\sum_{\ell\in\mathcal T}a_{i\ell}$ and $(H_F)_{ij}=-a_{ij}$ for $i\ne j$, while $(L_{FL})_{i\ell}=-a_{i\ell}$. Since each row of $\mathcal{L}_{\mathcal G}$ sums to zero, $H_F\mathbf 1_M+L_{FL}\mathbf 1_N=0$. The digraph \(\mathcal G\) is said to have a united spanning tree rooted at the leaders if every follower is reachable from at least one leader.

\begin{assumption}\label{Assumption: Graph}
The digraph $\mathcal{G}$ has a united spanning tree rooted at the leaders.
\end{assumption}

Under Assumption~\ref{Assumption: Graph}, $H_F$ is a nonsingular M-matrix \cite{neumann1980m}. Consequently, $H_F^{-1}\ge 0$, $-H_F^{-1}L_{FL}\ge0$, and $-H_F^{-1}L_{FL}\mathbf 1_N=\mathbf 1_M$, which implies that $-H_F^{-1}L_{FL}$ is row stochastic \cite{cao2012distributed}. It should be noted that the existence of a united spanning tree rooted at the leaders is a necessary topological condition for containment tracking \cite{cao2012distributed}.

\subsection{Followers' Dynamics}
Each follower in this work is modeled as a linear MIMO system with a control input $u_i \in \mathbb{R}^{m}$ and an output $y_{i}\in \mathbb{R}^{m}$ with a vector relative degree denoted by $\mathbf{d}_i = [d_1^i, \ldots, d_m^i]^\top$, with respect to the output $y_{i}$. The followers' dynamics are expressed in the normal-form coordinates \cite{mueller2009normal} as,
\begin{equation}\label{eq: follower_dynamics_lcss}
\begin{cases}
\dot{x}_i=A_ix_i+B_i(L_i\eta_i+\psi_i u_i)\\
\dot{\eta}_i=\Gamma_i \eta_i+\Lambda_i y_{i}\\
y_{i}=C_ix_i
\end{cases}
\end{equation}
where the external-state $x_i\in\mathbb{R}^{\sum_{k=1}^m d_k^i}$ contains the outputs and their derivatives up to order \(d_k^i-1\) in each output channel and $\eta_i\in\mathbb{R}^{n_i-\sum_{k=1}^m d_k^i}$ contains the zero-dynamics coordinates. The matrices $A_i$ and $B_i$ are in the normal controllable canonical form; $C_i$ is a selector matrix that extracts the output $y_i$ from $x_i$; and $L_i$, $\Gamma_i$, and $\Lambda_i$ describe the coupling between the external dynamics and the zero dynamics and are known to designers. The matrix $\psi_i\in\mathbb{R}^{m\times m}$ is known and invertible. We refer the reader to \cite{mueller2009normal} for more details on the structure of matrices.

\begin{assumption}\label{Assumption: Minimum-phase}
The matrix \(\Gamma_i\) associated with the followers' zero dynamics is Hurwitz.
\end{assumption}

Since the followers do not have access to the leaders' model to embed a leader exosystem and solve the regulator equations for that model, the local tracking problem must remain solvable for every admissible bounded leader trajectory, which requires this minimum-phase condition. In this setup, the same condition also rules out unstable zero dynamics, which are known to create vulnerabilities to zero-dynamic attacks in cyber-physical systems \cite{weerakkody2016information}. 

\begin{assumption}\label{Assumption: Measurements}
The variables $x_i$ are available for control implementation.
\end{assumption}

For the present normal-form model, this assumption does not require measurement of the zero-dynamics state $\eta_i$, in contrast to full-state implementations commonly used in resilient containment designs \cite{xiao2020distributed, zuo2023resilient}. An output-feedback extension can be obtained under the standard regularity assumptions required by exact differentiators \cite{oliveira2018generalized, oliveira2017global}, but this extension is not pursued here. 

\subsection{Leaders' Dynamics}
For each leader $\ell\in\mathcal T$, the state $\phi_\ell$ evolves according to
\begin{equation}\label{Eq: leader-dynamics}
\dot{\phi}_\ell(t)=\nu_\ell(t),\quad \phi_\ell(t)\in\mathbb R^m
\end{equation}

\begin{assumption}\label{Assumption: Leaders bound}
For each leader $\ell\in\mathcal T$, the signal $\phi_\ell:[0,\infty)\to\mathbb R^m$ is defined for all $t\ge0$ and is locally absolutely continuous. Moreover, for every admissible realization of the leader motion, there exist finite constants $\bar\phi_\ell>0$ and $\bar\nu_\ell>0$, possibly depending on the realized trajectory and its initial condition, such that $\|\phi_\ell(t)\|\le \bar\phi_\ell,\quad \|\dot\phi_\ell(t)\|\le \bar\nu_\ell$ for almost all $t\ge0$. These constants, however, are unknown. 
\end{assumption}

Assumption~\ref{Assumption: Leaders bound} is imposed at the trajectory level. It does not require followers to know a leader model, an exosystem realization, or bounds on motion. It includes trajectories generated by the classical marginally stable LTI exosystems $\dot\omega_\ell=S_\ell\omega_\ell$, $\phi_\ell=C_\ell\omega_\ell$, whenever the eigenvalues of $S_\ell$ on the imaginary axis are semisimple. More generally, nonlinear exosystems are also covered whenever the realized trajectory remains in a compact set, and the generated output and its derivative are bounded along that trajectory.

The information received from a leader-neighbor is limited to the instantaneous task-space signal \(\phi_\ell(t)\), or equivalently to the relative signal used in the distributed protocol. No follower is assumed to know \(\nu_\ell(t)\), a bound on \(\|\nu_\ell\|\), a leader exosystem realization, a motion envelope, or any future value of \(\phi_\ell\). The boundedness constants in Assumption~\ref{Assumption: Leaders bound} are used only for analysis and are unavailable for tuning the controller parameters.

With $\Phi=\operatorname{col}(\phi_1,\ldots,\phi_N), \Phi_L(t)=\{\phi_\ell(t):\ell\in\mathcal T\}$ and $\nu_L=\dot\Phi$, one has $\Phi\in W^{1,\infty}([0,\infty);\mathbb R^{Nm})$ and $\nu_L\in L_\infty([0,\infty);\mathbb R^{Nm})$, with unknown finite essential bounds. We also denote by $\operatorname{co}(\Phi_L(t))$ the convex hull of the leaders' positions at each instant, with point-to-set distance defined as $\operatorname{dist}(y, S)=\inf_{q\in S}\|y-q\|$.

\subsection{Actuator Attacks}
Followers are subject to actuator attacks modeled as
\begin{equation}\label{Eq: AttackActuator}
u_i(t)=u_{i,\mathrm{nom}}(t)+w_{a_i}(t),
\qquad
w_{a_i}(t)=w_{a_i}^{c}(t)+w_{a_i}^{uc}(t).
\end{equation}
The correlated component $w_{a_i}^{c}$ is generated by
\begin{equation}\label{Eq:CorrelatedAttackRKHS}
w_{a_i}^{c}(t)
= K_{a_i}^x(x_i,t)x_i(t) + K_{a_i}^u(t)u_{i,\mathrm{nom}}(t),
\end{equation}
where \(K_{a_i}^x(x_i,t)x_i(t)\) denotes the state-correlated actuator-side component and \(K_{a_i}^u(t)u_{i,\text{nom}}(t)\) denotes the input-correlated part. The term \(w_{a_i}^{uc}(t)\) is an exogenous false-data injection.

\begin{assumption}\label{Assumption: state-correlated attack envelope}
For each follower \(i\), the state-correlated attack coefficient \(K_{a_i}^x(x_i,t)\) is Carathéodory in \((x_i,t)\) and locally Lipschitz in \(x\) on compact sets, uniformly on finite time intervals up to a locally integrable Lipschitz modulus. Moreover, there exist an unknown constant \(\kappa_x^i>0\) and a known envelope \(\omega_{x_i}:\mathbb R^{\sum_{k=1}^m d_k^i}\times[0,\infty)\to[0,\infty)\) such that \(\omega_{x_i}\) is continuous, locally Lipschitz in \(x\), and satisfies, for every \(R>0\),
\begin{equation}
\sup_{\|x\|\le R,\ t\ge0}\omega_{x_i}(x,t)<\infty .
\end{equation}
Furthermore,
\begin{equation}\label{Eq:StateAttackEnvelopeBound}
\|K_{a_i}^x(x_i,t)\|
\le
\kappa_x^i\omega_{x_i}(x_i,t)
\end{equation}
for all \(x\) and for almost all \(t\ge0\).
\end{assumption}
Assumption~\ref{Assumption: state-correlated attack envelope} specifies the actuator-attack classes for which recovery guarantees are sought. The defender does not need to know the magnitude of the state-correlated actuator attack, but must choose a known envelope $\omega_{x_i}$ that upper-bounds the growth of the presumed cyber-attacks. This reflects the standard adaptive-robust-control tradeoffs, where unknown uncertainty magnitudes are handled by adaptive gains, whereas the admissible growth structure must be specified through a known envelope \cite{wang2025robust}. 

The envelope \(\omega_{x_i}\) can be selected from a prescribed admissible class for the state-correlated actuator attack. One convenient way to obtain such an envelope is through a reproducing-kernel Hilbert space (RKHS) description \cite{wang2025robust, oesterheld2023model}. Briefly, an RKHS \(\mathcal H_{x_i}\) is a Hilbert space of functions associated with a positive definite kernel \(k_{x_i}\), where point evaluations are represented by the kernel. Thus, for any \(f\in\mathcal H_{x_i}\), we have $f(x)=\langle f,k_{x_i}(\cdot,x)\rangle_{\mathcal H_{x_i}}$, and then Cauchy-Schwarz inequality gives $|f(x)|\le \|f\|_{\mathcal H_{x_i}}\sqrt{k_{x_i}(x,x)}$. Therefore, if each scalar entry of \(K_{a_i}^x(\cdot,t)\) belongs to \(\mathcal H_{x_i}\) for almost all \(t\), with an unknown finite essential supremum of its RKHS norm over time, then Assumption~\ref{Assumption: state-correlated attack envelope} is satisfied, after absorbing fixed dimension-dependent constants into \(\kappa_x^i\), with the known envelope $\omega_{x_i}(x,t)=\sqrt{k_{x_i}(x,x)}$. If \(k_{x_i}\) is a Gaussian kernel, then \(\omega_{x_i}\equiv1\), and the usual bounded state-correlated attack gain is recovered.

\begin{assumption}\label{Assumption: Input Correlated Attack}
For each follower \(i\), the map \(K_{a_i}^u:[0,\infty)\to\mathbb R^{m\times m}\) is continuous and essentially bounded satisfying $\sup_{t\ge0}\big\|K_{a_i}^u(t)\big\|\le \overline{k_u^i}$ for some unknown $\overline{k_u^i}$. Denoting $\Delta_{u_i}(t)=I_m+\psi_iK_{a_i}^u(t)\psi_i^{-1}$, we further assume that there exists an unknown constant \(\chi_i>0\) such that $ s^\top\Delta_{u_i}(t)s\ge \chi_i\|s\|^2$ for all \(s\in\mathbb R^m\) and for almost all \(t\ge0\).
\end{assumption}

\begin{assumption}\label{Assumption: Uncorrelated attack}
For each follower \(i\), the exogenous false-data injection \(w_{a_i}^{uc}(t)\) is continuous and essentially bounded, i.e., there exists an unknown finite constant \(k_w^i>0\) such that $\sup_{t\ge0}\big\|w_{a_i}^{uc}(t)\big\|\le k_w^i$.
\end{assumption}

The condition in Assumption~\ref{Assumption: Input Correlated Attack} preserves the defender's strictly positive effective input authority and prevents the input-correlated attack from canceling the defender's command direction. 

For example, in a UAV application, an actuator-side adversary may attempt to drive the vehicle toward an adversarial reference trajectory and hijack the vehicle while partially rejecting the defender's nominal corrections. If the adversarial reference is generated by a bounded marginally stable exosystem, its feedforward contribution can be represented through \(w_{a_i}^{uc}(t)\). The feedback component induced by the compromised actuator channel can then be represented by \(K_{a_i}^x(x_i,t)x_i\) that can be an unknown nonlinear function of the disclosed variables \(x_i\), with an admissible amplitude class that can be bounded, polynomial-growth, or belonging to a prescribed kernel-induced class. Destabilizing actuator attacks \cite{nematollahi2025cyber} can also be captured by this parameterization whenever their dependence on \(x_i\) admits a known envelope of the form required in Assumption~\ref{Assumption: state-correlated attack envelope}.

With the system description and attack models in place, we now state the control objective. In this setup, we distinguish command-level containment from physical-output containment. In fact, for followers with heterogeneous higher-order relative degrees, exact physical-output tracking of arbitrary moving commands would generally require higher-order command derivatives. However, since the available leader information is limited to locally absolutely continuous trajectories with unknown velocity bounds, and under the information pattern considered in this paper, such information is nonexistent or unavailable. Accordingly, the physical outputs are required to converge to the leaders' convex hull, with a residual determined by the local command-tracking error, whereas the local commands need to converge asymptotically.

\begin{problem}\label{prob:main}
Consider the heterogeneous MAS \eqref{eq: follower_dynamics_lcss} over the graph \eqref{eq: Laplacian_partition}, with leaders \eqref{Eq: leader-dynamics} and actuator attacks \eqref{Eq: AttackActuator}-\eqref{Eq:CorrelatedAttackRKHS}. Under Assumptions~\ref{Assumption: Graph}-\ref{Assumption: Uncorrelated attack}, design continuous local nominal controllers \(u_{i,\mathrm{nom}}\) and a continuous distributed interaction law such that all closed-loop follower signals remain bounded and, for every follower \(i\in\mathcal F\),
\begin{equation}\label{eq:obj}
\limsup_{t\to\infty}
\operatorname{dist}\bigl(y_i(t),\operatorname{co}(\Phi_L(t))\bigr)
\le \varepsilon_i ,
\end{equation}
where \(\varepsilon_i\ge0\) is to be determined by the closed-loop architecture. The local controller may use only the external state \(x_i\) and local adaptive variables. The distributed interaction law may use only neighbor-exchanged network-interface states and instantaneous leader-neighbor task-space signals when leader-neighbor edges exist.
\end{problem}

\section{Control Architecture and Main Results}\label{sec:control_stability}
\subsection{Command Filter, Immersion, and Error Dynamics}

For each follower \(i\in\mathcal F\), let \(r_i\in\mathbb R^m\) denote the task-space command supplied by the network-interface layer to be designed later, and choose a stable local command filter
\begin{equation}\label{eq:local_ref_gen}
\dot z_i=F_iz_i+G_ir_i,
\end{equation}
where \(F_i\) is Hurwitz. The filter \eqref{eq:local_ref_gen} is a local design object and is not a model of the leaders. It is called admissible for follower \(i\) if there exist matrices \(\Pi_{i_1}\), \(\Pi_{i_2}\), \(Q_i\), and \(K_i\), with \(\Pi_i=\operatorname{col}(\Pi_{i_1},\Pi_{i_2})\) and \(H_i=C_i\Pi_{i_1}\), such that
\begin{equation}\label{eq:embed}
\begin{split}
\Pi_iF_i
&=
\begin{bmatrix}
A_i&B_iL_i\\
\Lambda_iC_i&\Gamma_i
\end{bmatrix}\Pi_i
+
\begin{bmatrix}
B_i\\
0
\end{bmatrix}Q_i,\\
\Pi_iG_i
&=
\begin{bmatrix}
B_i\\
0
\end{bmatrix}K_i,\\
\operatorname{rank}H_i&=m,\qquad
\operatorname{rank}\left(H_iF_i^{-1}G_i\right)=m .
\end{split}
\end{equation}
The two matching equations in \eqref{eq:embed} ensure that every trajectory generated by the command filter can be embedded into the nominal follower dynamics. The rank condition on \(H_i\) ensures that the embedded nominal output spans the \(m\)-dimensional task space. The rank condition on \(H_iF_i^{-1}G_i\) ensures that the command-interface map used later is well defined.

We next show that admissible command filters are not an additional restrictive assumption. They can be constructed directly from the follower's normal-form representation. Let
\[
n_{x_i}=\sum_{k=1}^m d_k^i,\qquad
n_{\eta_i}=n_i-n_{x_i}.
\]
For each output channel \(k\), choose a Hurwitz polynomial
\[
p_{ik}(s)=s^{r_k^i}+a_{ik,r_k^i}s^{r_k^i-1}+\cdots+a_{ik,2}s+a_{ik,1},
\qquad a_{ik,1}\ne0 .
\]
Because \(A_i\) and \(B_i\) are in the normal controllable canonical form associated with the vector relative degree \([d_1^i,\ldots,d_m^i]^T\), one can choose an external-chain feedback matrix \(K_{x_i}\) such that
\begin{equation*}\label{eq:Fx_construction}
F_{x_i}=A_i+B_iK_{x_i}
\end{equation*}
has, in channel \(k\), the characteristic polynomial \(p_{ik}\). Hence \(F_{x_i}\) is Hurwitz. More explicitly, for the \(k\)-th chain, the last derivative can be assigned as
\[
\dot \xi_{ik,r_k^i}
=
-a_{ik,1}\xi_{ik,1}
-a_{ik,2}\xi_{ik,2}
-\cdots
-a_{ik,r_k^i}\xi_{ik,r_k^i}
+\tilde r_{ik},
\]
which gives the desired stable chain dynamics from \(\tilde r_{ik}\) to the embedded output \(\xi_{ik,1}\). Consequently,
\begin{equation*}\label{eq:D_i_dc_gain}
D_i=C_iF_{x_i}^{-1}B_i
\end{equation*}
is nonsingular. In the decoupled canonical-chain realization, one has
\[
D_i=-\operatorname{diag}\left(a_{i1,1}^{-1},\ldots,a_{im,1}^{-1}\right),
\]
which is nonsingular because \(a_{ik,1}\ne0\) for all \(k\).

Now choose a nonsingular matrix \(S_i\in\mathbb R^{n_{\eta_i}\times n_{\eta_i}}\) and define the internal filter realization
\begin{equation*}\label{eq:assigned_zero_filter}
F_{\eta_i}^{d}=S_i^{-1}\Gamma_iS_i,\qquad
\Lambda_i^{d}=S_i^{-1}\Lambda_i .
\end{equation*}
Since \(\Gamma_i\) is Hurwitz by Assumption~\ref{Assumption: Minimum-phase}, \(F_{\eta_i}^{d}\) is also Hurwitz. This construction explicitly assigns the zero-dynamics part of the command filter to a prescribed stable realization similar to the follower zero dynamics. In particular, taking \(S_i=I\) gives \(F_{\eta_i}^{d}=\Gamma_i\). This is the appropriate assignment freedom for the internal part, where the zero-dynamics eigenvalues are intrinsic to the plant normal form and cannot be arbitrarily reassigned by a command filter, but any stable coordinate realization similar to \(\Gamma_i\) can be used. Define
\begin{equation}\label{eq:filter_constructive_choice}
F_i=
\begin{bmatrix}
F_{x_i} & 0\\
\Lambda_i^{d}C_i & F_{\eta_i}^{d}
\end{bmatrix},
\qquad
G_i=
\begin{bmatrix}
B_iK_i\\
0
\end{bmatrix},
\end{equation}
where \(K_i\in\mathbb R^{m\times m}\) is chosen nonsingular, for example
\begin{equation}\label{eq:Ki_choice_filter}
K_i=-D_i^{-1}.
\end{equation}
Also choose
\begin{equation}\label{eq:Pi_Q_constructive_choice}
\Pi_{i_1}=\begin{bmatrix}I_{n_{x_i}}&0\end{bmatrix},\qquad
\Pi_{i_2}=\begin{bmatrix}0&S_i\end{bmatrix},\qquad
Q_i=\begin{bmatrix}K_{x_i}&-L_iS_i\end{bmatrix}.
\end{equation}
Then \(\Pi_i=\operatorname{col}(\Pi_{i_1},\Pi_{i_2})=\operatorname{diag}(I_{n_{x_i}},S_i)\). With \eqref{eq:filter_constructive_choice} and \eqref{eq:Pi_Q_constructive_choice},
\[
\Pi_iF_i
=
\begin{bmatrix}
I&0\\
0&S_i
\end{bmatrix}
\begin{bmatrix}
F_{x_i} & 0\\
S_i^{-1}\Lambda_iC_i & S_i^{-1}\Gamma_iS_i
\end{bmatrix}
=
\begin{bmatrix}
F_{x_i} & 0\\
\Lambda_iC_i & \Gamma_iS_i
\end{bmatrix}.
\]
On the other hand,
\[
\begin{split}
&\begin{bmatrix}
A_i&B_iL_i\\
\Lambda_iC_i&\Gamma_i
\end{bmatrix}
\Pi_i
+
\begin{bmatrix}
B_i\\
0
\end{bmatrix}Q_i
\\
&=
\begin{bmatrix}
A_i&B_iL_i\\
\Lambda_iC_i&\Gamma_i
\end{bmatrix}
\begin{bmatrix}
I&0\\
0&S_i
\end{bmatrix}
+
\begin{bmatrix}
B_i\\
0
\end{bmatrix}
\begin{bmatrix}
K_{x_i}&-L_iS_i
\end{bmatrix}.
\end{split}
\]
Therefore,
\[
\begin{bmatrix}
A_i+B_iK_{x_i} & B_iL_iS_i-B_iL_iS_i\\
\Lambda_iC_i & \Gamma_iS_i
\end{bmatrix}
=
\begin{bmatrix}
F_{x_i} & 0\\
\Lambda_iC_i & \Gamma_iS_i
\end{bmatrix}
=
\Pi_iF_i.
\]
Thus the first matching equation in \eqref{eq:embed} holds. Similarly,
\[
\Pi_iG_i
=
\begin{bmatrix}
I&0\\
0&S_i
\end{bmatrix}
\begin{bmatrix}
B_iK_i\\
0
\end{bmatrix}
=
\begin{bmatrix}
B_iK_i\\
0
\end{bmatrix}
=
\begin{bmatrix}
B_i\\
0
\end{bmatrix}K_i,
\]
which proves the second matching equation in \eqref{eq:embed}.

It remains to verify the rank conditions. Since
\[
H_i=C_i\Pi_{i_1}=\begin{bmatrix}C_i&0\end{bmatrix},
\]
and \(C_i\) extracts the \(m\) output coordinates from the external normal-form state, \(\operatorname{rank}H_i=m\). Moreover, \(F_i\) is block lower triangular with diagonal blocks \(F_{x_i}\) and \(F_{\eta_i}^{d}\), both Hurwitz. Hence \(F_i\) is Hurwitz. Since \(G_i=\operatorname{col}(B_iK_i,0)\), the block triangular inverse gives
\[
H_iF_i^{-1}G_i
=
C_iF_{x_i}^{-1}B_iK_i
=
D_iK_i.
\]
With the choice \eqref{eq:Ki_choice_filter}, this becomes
\[
H_iF_i^{-1}G_i=-I_m,
\]
and therefore
\[
\operatorname{rank}\left(H_iF_i^{-1}G_i\right)=m.
\]
This proves that an admissible command filter satisfying \eqref{eq:embed} exists for every follower in normal form under Assumption~\ref{Assumption: Minimum-phase}. The zero-dynamics part of the filter is explicitly included through \(F_{\eta_i}^{d}=S_i^{-1}\Gamma_iS_i\), which can be chosen as any stable coordinate realization of the plant zero dynamics.

Now, choose the nominal actuator command as
\begin{equation}\label{eq:unom}
u_{i,\mathrm{nom}}=\psi_i^{-1}\bigl(Q_iz_i+K_ir_i+u_i^r\bigr),
\end{equation}
where \(u_i^r\in\mathbb R^m\) is the recovery input to be designed. Define the local embedding errors
\begin{equation}\label{eq:local_errors}
e_{x_i}=x_i-\Pi_{i_1}z_i,\qquad
e_{\eta_i}=\eta_i-\Pi_{i_2}z_i.
\end{equation}
For compactness, set
\begin{equation*}\label{eq:Delta_d0_def}
\Delta_{x_i}(x_i,t)=\psi_iK_{a_i}^x(x_i,t),\qquad
d_{0_i}(t)=\psi_iw_{a_i}^{uc}(t).
\end{equation*}
Using \eqref{eq:unom}, the actuator attack model \eqref{Eq: AttackActuator}--\eqref{Eq:CorrelatedAttackRKHS}, and the immersion identities \eqref{eq:embed}, the error dynamics are
\begin{equation}\label{eq:errdyn}
\begin{split}
\dot e_{x_i}={}&(A_i+B_i\Delta_{x_i})e_{x_i}
+B_iL_ie_{\eta_i}
+B_i\Delta_{u_i}u_i^r\\
&+B_i(w_{\kappa_i}+w_{r_i}+d_{0_i}),\\
\dot e_{\eta_i}={}&\Gamma_ie_{\eta_i}+\Lambda_iC_ie_{x_i},
\end{split}
\end{equation}
where
\begin{equation}\label{eq:separated_terms}
\begin{split}
w_{\kappa_i}
&=(\Delta_{x_i}\Pi_{i_1}+\psi_iK_{a_i}^u\psi_i^{-1}Q_i)z_i,\\
w_{r_i}
&=\psi_iK_{a_i}^u\psi_i^{-1}K_ir_i.
\end{split}
\end{equation}
Indeed, if \(w_{a_i}=0\), \(u_i^r=0\), and the follower is initialized on the manifold \(x_i=\Pi_{i_1}z_i\), \(\eta_i=\Pi_{i_2}z_i\), then \eqref{eq:embed} makes this manifold invariant.

We next have the following lemma.
\begin{lemma}\label{lem:normalformdesign}
For each follower satisfying Assumption~\ref{Assumption: Minimum-phase}, fix any \(\Sigma_i=\Sigma_i^T>0\) and \(0<\varepsilon_{a_i}<\lambda_{\min}(\Sigma_i)\). Then there exists \(\bar\alpha_i>0\) such that, for all \(\alpha_i\ge\bar\alpha_i\), the Riccati equation
\begin{equation}\label{eq:are}
A_i^TP_i+P_iA_i-\alpha_iP_iB_iB_i^TP_i+\Sigma_i=0
\end{equation}
has a stabilizing solution \(P_i=P_i^T>0\), and there exist \(R_i=R_i^T>0\) and \(\varsigma_i>0\) such that
\begin{equation}\label{eq:R_lyap}
\Gamma_i^TR_i+R_i\Gamma_i=-I
\end{equation}
and
\begin{equation}\label{eq:Xi}
\Xi_i=
\begin{bmatrix}
\Sigma_i-\varepsilon_{a_i}I & -M_i\\
-M_i^T & \varsigma_i I
\end{bmatrix}>0,
\qquad
M_i=P_iB_iL_i+\varsigma_iC_i^T\Lambda_i^TR_i .
\end{equation}
\end{lemma}

\begin{proof}
Since \((A_i,B_i)\) is controllable, \eqref{eq:are} admits a stabilizing solution \(P_i(\alpha_i)=P_i^T(\alpha_i)>0\) for sufficiently large \(\alpha_i\). Moreover, by the cheap-control Riccati scaling for controllable normal-form chains \cite{saberi2000control},
\begin{equation}\label{eq:cheap_control_PB}
P_i(\alpha_i)B_i\to0
\qquad
\text{as } \alpha_i\to\infty .
\end{equation}
Consequently,
\begin{equation}\label{eq:PBIL_small}
P_i(\alpha_i)B_iL_i\to0
\qquad
\text{as } \alpha_i\to\infty .
\end{equation}

Since \(\Gamma_i\) is Hurwitz by Assumption~\ref{Assumption: Minimum-phase}, the Lyapunov equation \eqref{eq:R_lyap} has a unique solution \(R_i=R_i^T>0\). Define
\begin{equation*}\label{eq:lemma_aux_defs}
A_{0_i}=\Sigma_i-\varepsilon_{a_i}I,\qquad
A_{c_i}=P_iB_iL_i,\qquad
B_{c_i}=C_i^T\Lambda_i^TR_i .
\end{equation*}
Then \(A_{0_i}>0\), and by \eqref{eq:PBIL_small},
\begin{equation*}\label{eq:Ac_small_again}
a_i=\|A_{c_i}\|\to0
\qquad
\text{as } \alpha_i\to\infty .
\end{equation*}
With these definitions, $M_i=A_{c_i}+\varsigma_iB_{c_i}$. By the Schur complement, \(\Xi_i>0\) is equivalent to
\begin{equation}\label{eq:Schur_condition_Xi}
\varsigma_i I-(A_{c_i}+\varsigma_iB_{c_i})^TA_{0_i}^{-1}(A_{c_i}+\varsigma_iB_{c_i})>0 .
\end{equation}
For every \(\varsigma_i>0\),
\begin{equation*}\label{eq:Schur_norm_bound}
\begin{split}
&(A_{c_i}+\varsigma_iB_{c_i})^TA_{0_i}^{-1}(A_{c_i}+\varsigma_iB_{c_i})\\
&\le
\|A_{0_i}^{-1}\|
\|A_{c_i}+\varsigma_iB_{c_i}\|^2 I\\
&\le
\|A_{0_i}^{-1}\|
\bigl(a_i+\varsigma_i\|B_{c_i}\|\bigr)^2 I .
\end{split}
\end{equation*}
If \(a_i>0\), choose \(\varsigma_i=\sqrt{a_i}\). Then
\[
\|A_{0_i}^{-1}\|
\bigl(a_i+\sqrt{a_i}\|B_{c_i}\|\bigr)^2
=
O(a_i)
=
o(\sqrt{a_i})
=
o(\varsigma_i)
\]
as \(\alpha_i\to\infty\). Hence, for sufficiently large \(\alpha_i\), the positive term \(\varsigma_i I\) dominates the Schur-complement correction in \eqref{eq:Schur_condition_Xi}. If \(a_i=0\), then
\[
\|A_{0_i}^{-1}\|
\bigl(\varsigma_i\|B_{c_i}\|\bigr)^2
=
O(\varsigma_i^2),
\]
and any sufficiently small \(\varsigma_i>0\) makes \eqref{eq:Schur_condition_Xi} hold. Therefore, for all sufficiently large \(\alpha_i\), there exists \(\varsigma_i>0\) such that \(\Xi_i>0\).
\end{proof}

\subsection{Local Virtual-Actuator Recovery}
The local recovery layer is designed to ensure convergence of $e_{x_i}$ and $e_{\eta_i}$ despite the admissible class of actuator attacks. Let us set
\[
s_i=B_i^TP_ie_{x_i},\qquad
\kappa_{\psi_i}=\|\psi_i\|\|\psi_i^{-1}\|,
\]
and write \(\omega_{x_i}=\omega_{x_i}(x_i,t)\) for compactness. Let us define the known nonnegative regressors,
\begin{equation}\label{eq:Omega1}
\begin{split}
\Omega_{i,1}^{(1)}
&=
2\|\psi_i\|\|\Pi_{i_1}z_i\|\omega_{x_i},\\
\Omega_{i,2}^{(1)}
&=
2\kappa_{\psi_i}\|Q_iz_i\|,\\
\Omega_{i,3}^{(1)}
&=
2\kappa_{\psi_i}\|K_ir_i\|,\\
\Omega_{i,4}^{(1)}
&=
2\|\psi_i\|,\\
\Omega_i^{(2)}
&=
1+\|\psi_i\|^2\omega_{x_i}^2 .
\end{split}
\end{equation}
The following lemma collects the algebraic consequences of Assumptions~\ref{Assumption: state-correlated attack envelope}-\ref{Assumption: Uncorrelated attack} in the direction of \(s_i\).

\begin{lemma}\label{lem:attack_regressor_envelope}
Under Assumptions~\ref{Assumption: state-correlated attack envelope}--\ref{Assumption: Uncorrelated attack}, there exist unknown finite constants \(\Theta_i>0\) and \(\theta_{i,q}>0\), \(q=1,\ldots,4\), such that, for all admissible trajectories and for almost all \(t\ge0\),
\begin{equation}\label{eq:statebound_compact}
2s_i^T\Delta_{x_i}e_{x_i}
\le
\varepsilon_{a_i}\|e_{x_i}\|^2+
\frac{\|\psi_i\|^2(\kappa_x^i)^2}{\varepsilon_{a_i}}\omega_{x_i}^2\|s_i\|^2,
\end{equation}
and
\begin{equation}\label{eq:uncertainty_envelope_bound}
\begin{split}
&2s_i^T(w_{\kappa_i}+w_{r_i}+d_{0_i})
+\alpha_i\|s_i\|^2
+\frac{\|\psi_i\|^2(\kappa_x^i)^2}{\varepsilon_{a_i}}\omega_{x_i}^2\|s_i\|^2\\
&\le
\Theta_i\Omega_i^{(2)}\|s_i\|^2
+
\sum_{q=1}^4\theta_{i,q}\Omega_{i,q}^{(1)}\|s_i\|.
\end{split}
\end{equation}
\end{lemma}

\begin{proof}
By \eqref{Eq:StateAttackEnvelopeBound} and the definition \(\Delta_{x_i}=\psi_iK_{a_i}^x\),
\[
\|\Delta_{x_i}(x_i,t)\|
=
\|\psi_iK_{a_i}^x(x_i,t)\|
\le
\|\psi_i\|\kappa_x^i\omega_{x_i}(x_i,t).
\]
Therefore, by Cauchy--Schwarz and Young's inequality,
\[
\begin{split}
2s_i^T\Delta_{x_i}e_{x_i}
&\le
2\|s_i\|\|\Delta_{x_i}\|\|e_{x_i}\|\\
&\le
2\|\psi_i\|\kappa_x^i\omega_{x_i}\|s_i\|\|e_{x_i}\|\\
&\le
\varepsilon_{a_i}\|e_{x_i}\|^2+
\frac{\|\psi_i\|^2(\kappa_x^i)^2}{\varepsilon_{a_i}}\omega_{x_i}^2\|s_i\|^2,
\end{split}
\]
which proves \eqref{eq:statebound_compact}.

Next, using \eqref{eq:separated_terms}, the first component of \(w_{\kappa_i}\) satisfies
\begin{equation*}\label{eq:matched_bound_1}
\begin{split}
2s_i^T\Delta_{x_i}\Pi_{i_1}z_i
&\le
2\|s_i\|\|\Delta_{x_i}\|\|\Pi_{i_1}z_i\|\\
&\le
2\|s_i\|\|\psi_i\|\kappa_x^i\omega_{x_i}\|\Pi_{i_1}z_i\|\\
&=
\kappa_x^i\Omega_{i,1}^{(1)}\|s_i\|.
\end{split}
\end{equation*}
By Assumption~\ref{Assumption: Input Correlated Attack},
\[
\|K_{a_i}^u(t)\|\le \overline{k_u^i}.
\]
Hence the second component of \(w_{\kappa_i}\) satisfies
\begin{equation*}\label{eq:matched_bound_2}
\begin{split}
2s_i^T\psi_iK_{a_i}^u\psi_i^{-1}Q_iz_i
&\le
2\|s_i\|\|\psi_i\|\|K_{a_i}^u\|\|\psi_i^{-1}\|\|Q_iz_i\|\\
&\le
\overline{k_u^i}\Omega_{i,2}^{(1)}\|s_i\|.
\end{split}
\end{equation*}
Similarly, the \(r_i\)-dependent input-correlated term satisfies
\begin{equation*}\label{eq:matched_bound_3}
\begin{split}
2s_i^T\psi_iK_{a_i}^u\psi_i^{-1}K_ir_i
&\le
2\|s_i\|\|\psi_i\|\|K_{a_i}^u\|\|\psi_i^{-1}\|\|K_ir_i\|\\
&\le
\overline{k_u^i}\Omega_{i,3}^{(1)}\|s_i\|.
\end{split}
\end{equation*}
Finally, by Assumption~\ref{Assumption: Uncorrelated attack},
\[
\|w_{a_i}^{uc}(t)\|\le k_w^i,
\]
and therefore
\begin{equation*}\label{eq:matched_bound_4}
2s_i^Td_{0_i}
=
2s_i^T\psi_iw_{a_i}^{uc}
\le
2\|s_i\|\|\psi_i\|k_w^i
=
k_w^i\Omega_{i,4}^{(1)}\|s_i\|.
\end{equation*}

Combining these bounds gives
\[
2s_i^T(w_{\kappa_i}+w_{r_i}+d_{0_i})
\le
\sum_{q=1}^4\theta_{i,q}\Omega_{i,q}^{(1)}\|s_i\|,
\]
where one admissible choice is
\[
\theta_{i,1}=\kappa_x^i,\qquad
\theta_{i,2}=\overline{k_u^i},\qquad
\theta_{i,3}=\overline{k_u^i},\qquad
\theta_{i,4}=k_w^i.
\]
It remains to combine the purely quadratic terms. Since
\[
\Omega_i^{(2)}=1+\|\psi_i\|^2\omega_{x_i}^2,
\]
we have
\[
\alpha_i\|s_i\|^2
+
\frac{\|\psi_i\|^2(\kappa_x^i)^2}{\varepsilon_{a_i}}\omega_{x_i}^2\|s_i\|^2
\le
\Theta_i\Omega_i^{(2)}\|s_i\|^2
\]
for any finite constant satisfying
\[
\Theta_i
\ge
\max\left\{
\alpha_i,\frac{(\kappa_x^i)^2}{\varepsilon_{a_i}}
\right\}.
\]
Adding the linear and quadratic bounds proves \eqref{eq:uncertainty_envelope_bound} and completes the proof of lemma.
\end{proof}

Consider now the continuous virtual-actuator law,
\begin{equation}\label{eq:va}
\begin{cases}
u_i^r
=-\displaystyle\sum_{q=1}^4\dfrac{\beta_{i,q}^2(\Omega_{i,q}^{(1)})^2s_i} {\beta_{i,q}\Omega_{i,q}^{(1)}\|s_i\|+\mu_i(t)} -\dfrac{1}{2}\rho_i\Omega_i^{(2)}s_i,\\[3mm]
\dot\beta_{i,q}
=
-b_{\beta i,q}\mu_i(t)\beta_{i,q}
+b_{\beta i,q}\Omega_{i,q}^{(1)}\|s_i\|,
\quad q=1,\ldots,4,\\[1mm]
\dot\rho_i
=
-b_{\rho i}\mu_i(t)\rho_i
+b_{\rho i}\Omega_i^{(2)}\|s_i\|^2,
\end{cases}
\end{equation}
where \(\beta_{i,q}(0)>0\), \(\rho_i(0)>0\), \(b_{\beta i,q}>0\), \(b_{\rho i}>0\), and \(\mu_i:[0,\infty)\to(0,\infty)\) is continuous and satisfies \(\mu_i\in L_1[0,\infty)\).

\begin{theorem}\label{thm:local}
Under Assumptions~\ref{Assumption: Minimum-phase}, \ref{Assumption: Measurements}, and~\ref{Assumption: state-correlated attack envelope}--\ref{Assumption: Uncorrelated attack}, let the local command filter be admissible. Choose \(P_i\), \(R_i\), and \(\varsigma_i\) as in Lemma~\ref{lem:normalformdesign} with \(\alpha_i\ge\bar\alpha_i\). If
\[
r_i\in C([0,\infty),\mathbb R^m)\cap L_\infty[0,\infty),
\]
then, for every initial condition satisfying \(\beta_{i,q}(0)>0\) and \(\rho_i(0)>0\), the local closed-loop system \eqref{eq:local_ref_gen}, \eqref{eq:errdyn}, and \eqref{eq:va} has a unique complete Carathéodory solution. All local signals are bounded, \(u_i^r\) is continuous along the system trajectory, and
\begin{equation}\label{eq:local_conv_statement}
\lim_{t\to\infty}e_{x_i}(t)=0,\qquad
\lim_{t\to\infty}e_{\eta_i}(t)=0 .
\end{equation}
\end{theorem}

\begin{proof}
Consider the comparison inequalities for $q=1,\ldots,4$
\[
\dot\beta_{i,q}
=
-b_{\beta i,q}\mu_i(t)\beta_{i,q}
+
b_{\beta i,q}\Omega_{i,q}^{(1)}\|s_i\|
\ge
-b_{\beta i,q}\mu_i(t)\beta_{i,q},
\]
and
\[
\dot\rho_i
=
-b_{\rho i}\mu_i(t)\rho_i
+
b_{\rho i}\Omega_i^{(2)}\|s_i\|^2
\ge
-b_{\rho i}\mu_i(t)\rho_i.
\]
With positive initial conditions, these inequalities imply that
\[
\beta_{i,q}(t)>0,\qquad \rho_i(t)>0
\]
on every finite interval on which the solution exists. Hence, the denominators in \eqref{eq:va} are strictly positive. Since \(r_i\) and \(\mu_i\) are continuous, and since \(K_{a_i}^x\) satisfies the Carathéodory and local Lipschitz conditions in Assumption~\ref{Assumption: state-correlated attack envelope}, the augmented local closed-loop dynamics have a unique maximal Carathéodory solution on some interval \([0,T_{\max})\).

Because \(F_i\) is Hurwitz and \(r_i\in L_\infty[0,\infty)\), the command filter \eqref{eq:local_ref_gen} satisfies
\begin{equation*}\label{eq:z_bounded_local_proof}
z_i\in L_\infty[0,T_{\max})
\end{equation*}
on every maximal finite interval. Define
\[
\zeta_i=\operatorname{col}(e_{x_i},e_{\eta_i})
\]
and consider
\begin{equation}\label{eq:V0}
V_{0_i}=e_{x_i}^TP_ie_{x_i}+\varsigma_ie_{\eta_i}^TR_ie_{\eta_i}.
\end{equation}
Since \(P_i=P_i^T>0\), \(R_i=R_i^T>0\), and \(\varsigma_i>0\), \(V_{0_i}\) is positive definite in \(\zeta_i\). Differentiating \(V_{0_i}\) along \eqref{eq:errdyn} gives, for almost all \(t\in[0,T_{\max})\),
\begin{equation}\label{eq:V0_derivative_raw}
\begin{split}
\dot V_{0_i}
={}&
e_{x_i}^T(A_i^TP_i+P_iA_i)e_{x_i}
+
2e_{x_i}^TP_iB_iL_ie_{\eta_i}\\
&+
2s_i^T\Delta_{x_i}e_{x_i}
+
2s_i^T\Delta_{u_i}u_i^r\\
&+
2s_i^Tw_{\kappa_i}
+
2s_i^Tw_{r_i}
+
2s_i^Td_{0_i}\\
&+
\varsigma_ie_{\eta_i}^T(\Gamma_i^TR_i+R_i\Gamma_i)e_{\eta_i}
+
2\varsigma_ie_{\eta_i}^TR_i\Lambda_iC_ie_{x_i}.
\end{split}
\end{equation}
Using \eqref{eq:are},
\[
A_i^TP_i+P_iA_i
=
-\Sigma_i+\alpha_iP_iB_iB_i^TP_i,
\]
and using \(s_i=B_i^TP_ie_{x_i}\), we obtain
\begin{equation}\label{eq:ARE_term}
e_{x_i}^T(A_i^TP_i+P_iA_i)e_{x_i}
=
-e_{x_i}^T\Sigma_ie_{x_i}
+
\alpha_i\|s_i\|^2.
\end{equation}
Moreover, by \eqref{eq:R_lyap},
\begin{equation}\label{eq:zero_term}
\varsigma_ie_{\eta_i}^T(\Gamma_i^TR_i+R_i\Gamma_i)e_{\eta_i}
=
-\varsigma_i\|e_{\eta_i}\|^2,
\end{equation}
and
\begin{equation}\label{eq:cross_term}
\begin{split}
&2e_{x_i}^TP_iB_iL_ie_{\eta_i}
+
2\varsigma_ie_{\eta_i}^TR_i\Lambda_iC_ie_{x_i}\\
&=
2e_{x_i}^T
\left(
P_iB_iL_i+\varsigma_iC_i^T\Lambda_i^TR_i
\right)e_{\eta_i}
=
2e_{x_i}^TM_ie_{\eta_i}.
\end{split}
\end{equation}
Substituting \eqref{eq:ARE_term}--\eqref{eq:cross_term} into \eqref{eq:V0_derivative_raw} gives
\begin{equation}\label{eq:V0_derivative_expanded}
\begin{split}
\dot V_{0_i}
={}&
-e_{x_i}^T\Sigma_ie_{x_i}
+
2e_{x_i}^TM_ie_{\eta_i}
-
\varsigma_i\|e_{\eta_i}\|^2\\
&+
\alpha_i\|s_i\|^2
+
2s_i^T\Delta_{x_i}e_{x_i}
+
2s_i^T\Delta_{u_i}u_i^r\\
&+
2s_i^T(w_{\kappa_i}+w_{r_i}+d_{0_i}).
\end{split}
\end{equation}
By Lemma~\ref{lem:attack_regressor_envelope},
\[
2s_i^T\Delta_{x_i}e_{x_i}
\le
\varepsilon_{a_i}\|e_{x_i}\|^2+
\frac{\|\psi_i\|^2(\kappa_x^i)^2}{\varepsilon_{a_i}}\omega_{x_i}^2\|s_i\|^2.
\]
Using this inequality and \eqref{eq:Xi}, we obtain
\begin{equation}\label{eq:Xi_dissipation}
\begin{split}
&-e_{x_i}^T\Sigma_ie_{x_i}
+
2e_{x_i}^TM_ie_{\eta_i}
-
\varsigma_i\|e_{\eta_i}\|^2
+
\varepsilon_{a_i}\|e_{x_i}\|^2\\
&=
-\zeta_i^T
\begin{bmatrix}
\Sigma_i-\varepsilon_{a_i}I & -M_i\\
-M_i^T & \varsigma_iI
\end{bmatrix}
\zeta_i
=
-\zeta_i^T\Xi_i\zeta_i
\le
-c_i\|\zeta_i\|^2,
\end{split}
\end{equation}
where
\[
c_i=\lambda_{\min}(\Xi_i)>0.
\]
Combining \eqref{eq:V0_derivative_expanded}, \eqref{eq:Xi_dissipation}, and the envelope inequality \eqref{eq:uncertainty_envelope_bound} in Lemma~\ref{lem:attack_regressor_envelope}, we get
\begin{equation}\label{eq:V0dot_compact}
\dot V_{0_i}
\le
-c_i\|\zeta_i\|^2
+
2s_i^T\Delta_{u_i}u_i^r
+
\Theta_i\Omega_i^{(2)}\|s_i\|^2
+
\sum_{q=1}^4\theta_{i,q}\Omega_{i,q}^{(1)}\|s_i\|.
\end{equation}

Choose proof constants \(\beta_{i,q}^*>0\), \(q=1,\ldots,4\), and \(\rho_i^*>0\) such that
\begin{equation}\label{eq:ideal_gain_choice}
\chi_i\beta_{i,q}^*\ge\theta_{i,q},
\qquad
\chi_i\rho_i^*\ge\Theta_i.
\end{equation}
Let us define
\begin{equation}\label{eq:Va_adaptive_part}
V_{a_i}
=
\sum_{q=1}^{4}
\frac{\chi_i}{2b_{\beta i,q}}
(\beta_{i,q}-\beta_{i,q}^*)^2
+
\frac{\chi_i}{2b_{\rho i}}
(\rho_i-\rho_i^*)^2
\end{equation}
and
\begin{equation}\label{eq:Va}
\mathcal V_i=V_{0_i}+V_{a_i}.
\end{equation}
By Assumption~\ref{Assumption: Input Correlated Attack},
\begin{equation}\label{eq:input_authority_local}
s_i^T\Delta_{u_i}(t)s_i\ge\chi_i\|s_i\|^2.
\end{equation}
Using \eqref{eq:va}, the control-port term satisfies
\begin{equation}\label{eq:controlport_full}
\begin{split}
2s_i^T\Delta_{u_i}u_i^r
={}&
-2
\sum_{q=1}^4
\frac{\beta_{i,q}^2(\Omega_{i,q}^{(1)})^2}
{\beta_{i,q}\Omega_{i,q}^{(1)}\|s_i\|+\mu_i(t)}
s_i^T\Delta_{u_i}s_i\\
&-
\rho_i\Omega_i^{(2)}
s_i^T\Delta_{u_i}s_i\\
\le{}&
-2\chi_i
\sum_{q=1}^4
\frac{\beta_{i,q}^2(\Omega_{i,q}^{(1)})^2\|s_i\|^2}
{\beta_{i,q}\Omega_{i,q}^{(1)}\|s_i\|+\mu_i(t)}\\
&-
\chi_i\rho_i\Omega_i^{(2)}\|s_i\|^2.
\end{split}
\end{equation}
Differentiating \(V_{a_i}\) along the adaptive laws in \eqref{eq:va} gives
\begin{equation}\label{eq:Va_derivative_full}
\begin{split}
\dot V_{a_i}
={}&
\sum_{q=1}^{4}
\chi_i(\beta_{i,q}-\beta_{i,q}^*)
\left(
-\mu_i\beta_{i,q}
+
\Omega_{i,q}^{(1)}\|s_i\|
\right)\\
&+
\chi_i(\rho_i-\rho_i^*)
\left(
-\mu_i\rho_i
+
\Omega_i^{(2)}\|s_i\|^2
\right).
\end{split}
\end{equation}
Combining \eqref{eq:V0dot_compact}, \eqref{eq:controlport_full}, and \eqref{eq:Va_derivative_full}, we obtain
\begin{equation}\label{eq:Vdot_grouped_before_cancellation}
\begin{split}
\dot{\mathcal V}_i
\le{}&
-c_i\|\zeta_i\|^2\\
&+
\Theta_i\Omega_i^{(2)}\|s_i\|^2
-
\chi_i\rho_i\Omega_i^{(2)}\|s_i\|^2
+
\chi_i(\rho_i-\rho_i^*)\Omega_i^{(2)}\|s_i\|^2\\
&-
\chi_i\mu_i(\rho_i-\rho_i^*)\rho_i\\
&+
\sum_{q=1}^{4}
\Bigg[
\theta_{i,q}\Omega_{i,q}^{(1)}\|s_i\|
-
2\chi_i
\frac{\beta_{i,q}^2(\Omega_{i,q}^{(1)})^2\|s_i\|^2}
{\beta_{i,q}\Omega_{i,q}^{(1)}\|s_i\|+\mu_i}\\
&+\chi_i(\beta_{i,q}-\beta_{i,q}^*)\Omega_{i,q}^{(1)}\|s_i\|
-
\chi_i\mu_i(\beta_{i,q}-\beta_{i,q}^*)\beta_{i,q}
\Bigg].
\end{split}
\end{equation}
The quadratic terms satisfy
\begin{equation}\label{eq:rho_cancellation_full}
\begin{split}
&
\Theta_i\Omega_i^{(2)}\|s_i\|^2
-
\chi_i\rho_i\Omega_i^{(2)}\|s_i\|^2
+
\chi_i(\rho_i-\rho_i^*)\Omega_i^{(2)}\|s_i\|^2\\
&=
\left(\Theta_i-\chi_i\rho_i^*\right)\Omega_i^{(2)}\|s_i\|^2
\le 0,
\end{split}
\end{equation}
where the last inequality follows from \(\chi_i\rho_i^*\ge\Theta_i\). The leakage term generated by the \(\rho_i\)-adaptation satisfies
\begin{equation}\label{eq:rho_leakage_full}
\begin{split}
-\chi_i\mu_i(\rho_i-\rho_i^*)\rho_i
&=
-\chi_i\mu_i\rho_i^2
+
\chi_i\mu_i\rho_i^*\rho_i\\
&=
-\chi_i\mu_i
\left(\rho_i-\frac{\rho_i^*}{2}\right)^2
+
\frac{\chi_i(\rho_i^*)^2}{4}\mu_i\\
&\le
\frac{\chi_i(\rho_i^*)^2}{4}\mu_i.
\end{split}
\end{equation}

Next, fix any \(q\in\{1,\ldots,4\}\) and consider the corresponding \(\beta_{i,q}\)-dependent terms. First,
\begin{equation}\label{eq:beta_cancellation_start_full}
\begin{split}
&
\theta_{i,q}\Omega_{i,q}^{(1)}\|s_i\|
-
2\chi_i
\frac{\beta_{i,q}^2(\Omega_{i,q}^{(1)})^2\|s_i\|^2}
{\beta_{i,q}\Omega_{i,q}^{(1)}\|s_i\|+\mu_i}\\
&\qquad
+
\chi_i(\beta_{i,q}-\beta_{i,q}^*)\Omega_{i,q}^{(1)}\|s_i\|\\
&=
(\theta_{i,q}-\chi_i\beta_{i,q}^*)\Omega_{i,q}^{(1)}\|s_i\|
+
\chi_i\beta_{i,q}\Omega_{i,q}^{(1)}\|s_i\|\\
&\qquad
-
2\chi_i
\frac{\beta_{i,q}^2(\Omega_{i,q}^{(1)})^2\|s_i\|^2}
{\beta_{i,q}\Omega_{i,q}^{(1)}\|s_i\|+\mu_i}\\
&\le
\chi_i\beta_{i,q}\Omega_{i,q}^{(1)}\|s_i\|
-
2\chi_i
\frac{\beta_{i,q}^2(\Omega_{i,q}^{(1)})^2\|s_i\|^2}
{\beta_{i,q}\Omega_{i,q}^{(1)}\|s_i\|+\mu_i},
\end{split}
\end{equation}
because \(\theta_{i,q}-\chi_i\beta_{i,q}^*\le0\). Define
\begin{equation}\label{eq:xbeta_definition_full}
x_{i,q}(t)=\beta_{i,q}(t)\Omega_{i,q}^{(1)}(t)\|s_i(t)\|.
\end{equation}
Since \(\beta_{i,q}(t)>0\), \(\Omega_{i,q}^{(1)}(t)\ge0\), and \(\|s_i(t)\|\ge0\), we have \(x_{i,q}(t)\ge0\). Therefore,
\begin{equation}\label{eq:beta_fraction_full}
x_{i,q}
-
\frac{2x_{i,q}^2}{x_{i,q}+\mu_i}
=
\frac{x_{i,q}(\mu_i-x_{i,q})}{x_{i,q}+\mu_i}
\le
\mu_i.
\end{equation}
Indeed, if \(x_{i,q}\ge\mu_i\), then the left-hand side is nonpositive; if \(0\le x_{i,q}<\mu_i\), then \(x_{i,q}(\mu_i-x_{i,q})/(x_{i,q}+\mu_i)\le\mu_i\). Hence \eqref{eq:beta_cancellation_start_full} gives
\begin{equation}\label{eq:beta_cancellation_end_full}
\begin{split}
&
\theta_{i,q}\Omega_{i,q}^{(1)}\|s_i\|
-
2\chi_i
\frac{\beta_{i,q}^2(\Omega_{i,q}^{(1)})^2\|s_i\|^2}
{\beta_{i,q}\Omega_{i,q}^{(1)}\|s_i\|+\mu_i}\\
&\qquad
+
\chi_i(\beta_{i,q}-\beta_{i,q}^*)\Omega_{i,q}^{(1)}\|s_i\|
\le
\chi_i\mu_i.
\end{split}
\end{equation}
The leakage term generated by the \(\beta_{i,q}\)-adaptation satisfies
\begin{equation}\label{eq:beta_leakage_full}
\begin{split}
-\chi_i\mu_i(\beta_{i,q}-\beta_{i,q}^*)\beta_{i,q}
&=
-\chi_i\mu_i\beta_{i,q}^2
+
\chi_i\mu_i\beta_{i,q}^*\beta_{i,q}\\
&=
-\chi_i\mu_i
\left(\beta_{i,q}-\frac{\beta_{i,q}^*}{2}\right)^2
+
\frac{\chi_i(\beta_{i,q}^*)^2}{4}\mu_i\\
&\le
\frac{\chi_i(\beta_{i,q}^*)^2}{4}\mu_i.
\end{split}
\end{equation}

Substituting \eqref{eq:rho_cancellation_full}--\eqref{eq:beta_leakage_full} into \eqref{eq:Vdot_grouped_before_cancellation}, we obtain
\begin{equation}\label{eq:Vdotlocal}
\dot{\mathcal V}_i
\le
-c_i\|\zeta_i\|^2+\mathcal C_i\mu_i(t)
\end{equation}
for almost all \(t\in[0,T_{\max})\), where the finite constant
\begin{equation}\label{eq:C_local_full}
\mathcal C_i
=
4\chi_i
+
\frac{\chi_i}{4}\sum_{q=1}^{4}(\beta_{i,q}^*)^2
+
\frac{\chi_i(\rho_i^*)^2}{4}
\end{equation}
depends only on proof constants and unknown attack bounds, not on time.

Since \(\mu_i\in L_1[0,\infty)\), integration of \eqref{eq:Vdotlocal} over \([0,t]\subset[0,T_{\max})\) gives
\begin{equation}\label{eq:local_integral_bound}
\mathcal V_i(t)
+
c_i\int_0^t\|\zeta_i(\tau)\|^2d\tau
\le
\mathcal V_i(0)
+
\mathcal C_i\int_0^\infty\mu_i(\tau)d\tau .
\end{equation}
Therefore,
\begin{equation}\label{eq:V_zeta_adaptive_bounds_full}
\mathcal V_i\in L_\infty[0,T_{\max}),\qquad
\zeta_i\in L_2[0,T_{\max})\cap L_\infty[0,T_{\max}),
\end{equation}
and
\begin{equation}\label{eq:adaptive_bounds_full}
\beta_{i,q}\in L_\infty[0,T_{\max}),\qquad
\rho_i\in L_\infty[0,T_{\max}).
\end{equation}
Since \(z_i\in L_\infty[0,T_{\max})\), the relation
\[
x_i=e_{x_i}+\Pi_{i_1}z_i
\]
implies
\begin{equation}\label{eq:x_bound_local_full}
x_i\in L_\infty[0,T_{\max}).
\end{equation}
By Assumption~\ref{Assumption: state-correlated attack envelope}, \(\omega_{x_i}(x_i,t)\) is bounded on bounded \(x_i\)-sets uniformly in time. Hence
\begin{equation}\label{eq:omega_bound_local_full}
\omega_{x_i}(x_i,t)\in L_\infty[0,T_{\max}).
\end{equation}
Using \eqref{eq:Omega1}, together with \(z_i\in L_\infty\), \(r_i\in L_\infty\), and \eqref{eq:omega_bound_local_full}, we obtain
\begin{equation}\label{eq:Omega_bounds_local_full}
\Omega_{i,q}^{(1)}\in L_\infty[0,T_{\max}),\qquad
\Omega_i^{(2)}\in L_\infty[0,T_{\max}),\qquad q=1,\ldots,4.
\end{equation}
Furthermore, for each \(q\),
\begin{equation}\label{eq:va_fraction_bound_compact}
\frac{\beta_{i,q}^2(\Omega_{i,q}^{(1)})^2\|s_i\|}
{\beta_{i,q}\Omega_{i,q}^{(1)}\|s_i\|+\mu_i(t)}
\le
\beta_{i,q}\Omega_{i,q}^{(1)}.
\end{equation}
Indeed, if \(s_i=0\), the left-hand side is zero; otherwise the denominator is at least \(\beta_{i,q}\Omega_{i,q}^{(1)}\|s_i\|\). Therefore, by \eqref{eq:va}, \eqref{eq:adaptive_bounds_full}, and \eqref{eq:Omega_bounds_local_full},
\begin{equation}\label{eq:ur_bound_local_full}
u_i^r\in L_\infty[0,T_{\max}).
\end{equation}
From \eqref{eq:unom}, and using \(z_i,r_i,u_i^r\in L_\infty[0,T_{\max})\), we get
\begin{equation}\label{eq:unom_bound_local_full}
u_{i,\mathrm{nom}}\in L_\infty[0,T_{\max}).
\end{equation}
Assumptions~\ref{Assumption: Input Correlated Attack} and~\ref{Assumption: Uncorrelated attack}, together with \eqref{eq:x_bound_local_full}, \eqref{eq:omega_bound_local_full}, and \eqref{eq:unom_bound_local_full}, imply that
\[
\Delta_{x_i}(x_i,t),\qquad
\Delta_{u_i}(t),\qquad
d_{0_i}(t),\qquad
w_{\kappa_i},\qquad
w_{r_i}
\]
are bounded along the solution. Therefore the right-hand side of \eqref{eq:errdyn} is bounded on every bounded time interval contained in \([0,T_{\max})\). Also, the adaptive right-hand sides in \eqref{eq:va} are bounded on every such interval because \(\mu_i\) is continuous and the involved signals are bounded. Consequently, the maximal solution cannot escape to infinity in finite time. Since \(\beta_{i,q}(t)\) and \(\rho_i(t)\) remain strictly positive, the solution also cannot leave the domain on which \eqref{eq:va} is well defined. Hence
\[
T_{\max}=\infty.
\]

From \eqref{eq:errdyn}, write
\[
\dot\zeta_i
=
\begin{bmatrix}
(A_i+B_i\Delta_{x_i})e_{x_i}
+B_iL_ie_{\eta_i}
+B_i\Delta_{u_i}u_i^r
+B_i(w_{\kappa_i}+w_{r_i}+d_{0_i})\\
\Gamma_ie_{\eta_i}+\Lambda_iC_ie_{x_i}
\end{bmatrix}.
\]
All terms on the right-hand side are bounded along the complete solution. Hence
\begin{equation}\label{eq:zetadot_bound_local_full}
\dot\zeta_i\in L_\infty[0,\infty).
\end{equation}
From \eqref{eq:local_integral_bound}, we also have
\[
\zeta_i\in L_2[0,\infty).
\]
Applying Barbalat's lemma gives
\[
\lim_{t\to\infty}\zeta_i(t)=0.
\]
Therefore,
\[
\lim_{t\to\infty}e_{x_i}(t)=0,\qquad
\lim_{t\to\infty}e_{\eta_i}(t)=0.
\]

Finally, all local signals are bounded. Indeed, \(z_i\), \(e_{x_i}\), and \(e_{\eta_i}\) are bounded; hence \(x_i=e_{x_i}+\Pi_{i_1}z_i\) and \(\eta_i=e_{\eta_i}+\Pi_{i_2}z_i\) are bounded. The signals \(\beta_{i,q}\), \(\rho_i\), \(\Omega_{i,q}^{(1)}\), \(\Omega_i^{(2)}\), \(u_i^r\), and \(u_{i,\mathrm{nom}}\) are bounded by the preceding arguments. The proof is complete.
\end{proof}

\subsection{Network Interface and Command Containment}

For each follower \(i\in\mathcal F\), consider the network-interface protocol
\begin{equation}\label{eq:net}
\begin{cases}
\vartheta_i
=
\displaystyle\sum_{j\in\mathcal F}a_{ij}(\sigma_i-\sigma_j)
+
\sum_{\ell\in\mathcal T}a_{i\ell}(\sigma_i-\phi_\ell),\\[2mm]
\dot\gamma_i
=
-b_{\gamma i}\varpi_i(t)\gamma_i
+
b_{\gamma i}\|\vartheta_i\|,\\[1mm]
\dot\sigma_i
=
-\vartheta_i
-
\dfrac{\gamma_i^2\vartheta_i}
{\gamma_i\|\vartheta_i\|+\varpi_i(t)},
\end{cases}
\end{equation}
where \(\gamma_i(0)>0\), \(b_{\gamma i}>0\), and \(\varpi_i:[0,\infty)\to(0,\infty)\) is continuous and satisfies \(\varpi_i\in L_1[0,\infty)\).

\begin{theorem}\label{thm:network}
Under Assumptions~\ref{Assumption: Graph} and~\ref{Assumption: Leaders bound}, the protocol \eqref{eq:net} admits a unique complete Carathéodory solution for every initial condition with \(\gamma_i(0)>0\). Moreover, \(\sigma_i(t)\) is bounded and
\begin{equation}\label{eq:sigma_cont}
\lim_{t\to\infty}
\dist\bigl(\sigma_i(t),\co(\Phi_L(t))\bigr)=0,
\qquad i\in\mathcal F .
\end{equation}
\end{theorem}

\begin{proof}
We first establish well-posedness and positivity of the adaptive gains. Let
\[
\sigma=\operatorname{col}(\sigma_1,\ldots,\sigma_M),
\qquad
\gamma=\operatorname{col}(\gamma_1,\ldots,\gamma_M),
\]
and define the open domain
\[
\mathcal D=\mathbb R^{Mm}\times(0,\infty)^M .
\]
For each fixed \(t\), the right-hand side of \eqref{eq:net} is continuous in \((\sigma,\gamma)\in\mathcal D\). Since \(\varpi_i(t)>0\) and is continuous, the map
\[
(\sigma,\gamma)\mapsto
\frac{\gamma_i^2\vartheta_i}
{\gamma_i\|\vartheta_i\|+\varpi_i(t)}
\]
is locally Lipschitz on compact subsets of \(\mathcal D\), uniformly on finite time intervals. Moreover, Assumption~\ref{Assumption: Leaders bound} implies that the leader signals are locally absolutely continuous and locally bounded on finite intervals. Hence the right-hand side of \eqref{eq:net} satisfies the Carathéodory conditions and is locally Lipschitz in the state on compact subsets of \(\mathcal D\). Therefore, for every initial condition \(\sigma(0)\in\mathbb R^{Mm}\), \(\gamma_i(0)>0\), there exists a unique maximal Carathéodory solution on an interval \([0,T_{\max})\).

Along a solution, the gain equation can be written as
\[
\dot\gamma_i+b_{\gamma i}\varpi_i(t)\gamma_i
=
b_{\gamma i}\|\vartheta_i\|.
\]
The variation-of-constants formula gives, for \(0\le t<T_{\max}\),
\begin{equation}\label{eq:gamma_variation_formula}
\begin{split}
\gamma_i(t)
={}&
e^{-b_{\gamma i}\int_0^t\varpi_i(s)ds}\gamma_i(0)\\
&+
b_{\gamma i}
\int_0^t
e^{-b_{\gamma i}\int_s^t\varpi_i(r)dr}
\|\vartheta_i(s)\|ds .
\end{split}
\end{equation}
Both terms on the right-hand side are nonnegative, and the first one is strictly positive. Thus
\begin{equation}\label{eq:gamma_positive}
\gamma_i(t)>0,\qquad 0\le t<T_{\max}.
\end{equation}
This proves that the solution cannot leave \(\mathcal D\) through \(\gamma_i=0\).

Next define
\[
\Phi=\operatorname{col}(\phi_1,\ldots,\phi_N),
\qquad
\vartheta=\operatorname{col}(\vartheta_1,\ldots,\vartheta_M).
\]
From the graph partition,
\begin{equation}\label{eq:vartheta_stack}
\vartheta=(H_F\otimes I_m)\sigma+(L_{FL}\otimes I_m)\Phi .
\end{equation}
Let
\[
p=H_F^{-T}\mathbf 1_M .
\]
Under Assumption~\ref{Assumption: Graph}, \(H_F\) is a nonsingular \(M\)-matrix. Hence
\begin{equation}\label{eq:p_properties}
p_i>0,\qquad p^TH_F=\mathbf 1_M^T .
\end{equation}
Define
\[
g_\Phi(t)=(L_{FL}\otimes I_m)\dot\Phi(t).
\]
Assumption~\ref{Assumption: Leaders bound} gives \(g_\Phi\in L_\infty[0,\infty)\). Also define
\begin{equation}\label{eq:upsilon_def}
\upsilon_i
=
\frac{\gamma_i^2\vartheta_i}
{\gamma_i\|\vartheta_i\|+\varpi_i(t)},
\qquad
\upsilon=\operatorname{col}(\upsilon_1,\ldots,\upsilon_M).
\end{equation}
Differentiating \eqref{eq:vartheta_stack} along \eqref{eq:net} gives, for almost all \(t\),
\begin{equation}\label{eq:vartheta_dyn}
\dot\vartheta
=
-(H_F\otimes I_m)\vartheta
-
(H_F\otimes I_m)\upsilon
+
g_\Phi(t).
\end{equation}

Consider the nonsmooth gauge
\begin{equation}\label{eq:Jnet_def}
J(\vartheta)=\sum_{i=1}^Mp_i\|\vartheta_i\|.
\end{equation}
The function \(J\) is locally Lipschitz. Since \(\vartheta\) is absolutely continuous on every finite interval, \(\dot\vartheta(t)\) exists for almost all \(t\). For such \(t\), the upper Dini derivative of the Euclidean norm satisfies
\[
D^+\|\vartheta_i(t)\|
=
\max_{\xi_i\in\partial\|\vartheta_i(t)\|}
\xi_i^T\dot\vartheta_i(t),
\]
where
\[
\partial\|\vartheta_i\|=
\begin{cases}
\left\{\dfrac{\vartheta_i}{\|\vartheta_i\|}\right\}, & \vartheta_i\ne0,\\[2mm]
\{\xi\in\mathbb R^m:\|\xi\|\le1\}, & \vartheta_i=0.
\end{cases}
\]
Choose \(\xi_i(t)\in\partial\|\vartheta_i(t)\|\) to attain the maximum. Then
\[
\|\xi_i(t)\|\le1,\qquad
\xi_i^T(t)\vartheta_i(t)=\|\vartheta_i(t)\|,
\]
and
\begin{equation}\label{eq:Dini_J}
D^+J(\vartheta(t))
=
\sum_{i=1}^Mp_i\xi_i^T(t)\dot\vartheta_i(t)
\end{equation}
for almost all \(t\). We now estimate the three terms in \eqref{eq:vartheta_dyn}.

First, since \(H_F\) is an \(M\)-matrix, its off-diagonal entries satisfy \((H_F)_{ij}\le0\) for \(i\ne j\). Therefore, for each \(j\),
\[
\sum_{i=1}^Mp_i(H_F)_{ij}=1
\]
by \eqref{eq:p_properties}. By construction, \(\xi_j^T\vartheta_j=\|\vartheta_j\|\) for all \(j\), including the case \(\vartheta_j=0\), and \(\xi_i^T\vartheta_j\le\|\vartheta_j\|\) for all \(i,j\). Hence,
\begin{equation}\label{eq:nominal_gauge_bound}
\begin{split}
&\sum_{i=1}^Mp_i\xi_i^T[-(H_F\vartheta)_i]
=
-\sum_{i=1}^M\sum_{j=1}^Mp_i(H_F)_{ij}\xi_i^T\vartheta_j\\
&=
-\sum_{j=1}^M p_j(H_F)_{jj}\xi_j^T\vartheta_j
-\sum_{j=1}^M\sum_{\substack{i=1\\i\ne j}}^Mp_i(H_F)_{ij}\xi_i^T\vartheta_j\\
&\le
-\sum_{j=1}^M p_j(H_F)_{jj}\|\vartheta_j\|
+
\sum_{j=1}^M\sum_{\substack{i=1\\i\ne j}}^Mp_i[-(H_F)_{ij}]\|\vartheta_j\|\\
&=
-\sum_{j=1}^M
\left(
p_j(H_F)_{jj}
+
\sum_{\substack{i=1\\i\ne j}}^Mp_i(H_F)_{ij}
\right)\|\vartheta_j\|\\
&=
-\sum_{j=1}^M\|\vartheta_j\|.
\end{split}
\end{equation}
Second, define
\begin{equation}\label{eq:delta_net_def}
\delta_j=
\begin{cases}
\dfrac{\gamma_j^2\|\vartheta_j\|}
{\gamma_j\|\vartheta_j\|+\varpi_j(t)}, & \vartheta_j\ne0,\\[3mm]
0, & \vartheta_j=0.
\end{cases}
\end{equation}
Then \(\delta_j=\|\upsilon_j\|\), and if \(\vartheta_j\ne0\), \(\upsilon_j=\delta_j\xi_j\). Repeating the same \(M\)-matrix column argument gives
\begin{equation}\label{eq:robust_gauge_bound}
\sum_{i=1}^Mp_i\xi_i^T[-(H_F\upsilon)_i]
\le
-\sum_{j=1}^M\delta_j.
\end{equation}
Third, define
\begin{equation}\label{eq:Gphi_def}
G_\Phi
=
\operatorname*{ess\,sup}_{t\ge0}
\sum_{i=1}^Mp_i\|g_{\Phi i}(t)\|<\infty .
\end{equation}
Then
\begin{equation}\label{eq:leader_term_bound}
\sum_{i=1}^Mp_i\xi_i^Tg_{\Phi i}(t)
\le
\sum_{i=1}^Mp_i\|g_{\Phi i}(t)\|
\le
G_\Phi
\end{equation}
for almost all \(t\). Combining \eqref{eq:vartheta_dyn}--\eqref{eq:leader_term_bound}, we obtain
\begin{equation}\label{eq:J_derivative_main}
D^+J
\le
-\sum_{i=1}^M\|\vartheta_i\|
-\sum_{i=1}^M\delta_i
+
G_\Phi.
\end{equation}
Since
\[
J=\sum_{i=1}^Mp_i\|\vartheta_i\|
\le
p_{\max}\sum_{i=1}^M\|\vartheta_i\|,
\qquad
p_{\max}=\max_i p_i,
\]
we have
\[
\sum_{i=1}^M\|\vartheta_i\|\ge \frac{1}{p_{\max}}J.
\]
Let
\begin{equation}\label{eq:c0_definition}
c_0=\frac{1}{p_{\max}}>0.
\end{equation}
Then
\begin{equation}\label{eq:J_derivative_simplified}
D^+J
\le
-c_0J
-\sum_{i=1}^M\delta_i
+
G_\Phi .
\end{equation}

Set
\begin{equation}\label{eq:gamma_star_def}
\gamma^\star=G_\Phi+1
\end{equation}
and define
\begin{equation}\label{eq:Wnet}
W
=
\frac12J^2
+
\sum_{i=1}^M
\frac{p_i}{2b_{\gamma i}}(\gamma_i-\gamma^\star)^2 .
\end{equation}
Using \eqref{eq:J_derivative_simplified} and the adaptive law for \(\gamma_i\),
\begin{equation}\label{eq:Wdot_pre_net}
\begin{split}
D^+W
&=
JD^+J
+
\sum_{i=1}^M
p_i(\gamma_i-\gamma^\star)
\left(
-\varpi_i(t)\gamma_i+\|\vartheta_i\|
\right)\\
&\le
-c_0J^2
-
J\sum_{i=1}^M\delta_i
+
G_\Phi J\\
&\quad
+
\sum_{i=1}^Mp_i\gamma_i\|\vartheta_i\|
-
\gamma^\star\sum_{i=1}^Mp_i\|\vartheta_i\|\\
&\quad
-
\sum_{i=1}^Mp_i\varpi_i(t)(\gamma_i-\gamma^\star)\gamma_i .
\end{split}
\end{equation}
Since \(J=\sum_i p_i\|\vartheta_i\|\), the leader-bound term and the ideal gain term combine as
\[
G_\Phi J-\gamma^\star\sum_{i=1}^Mp_i\|\vartheta_i\|
=
-(\gamma^\star-G_\Phi)J
=
-J.
\]
Thus
\begin{equation}\label{eq:Wdot_after_gamma}
\begin{split}
D^+W
\le{}&
-c_0J^2
-J
-
J\sum_{i=1}^M\delta_i
+
\sum_{i=1}^Mp_i\gamma_i\|\vartheta_i\|\\
&-
\sum_{i=1}^Mp_i\varpi_i(t)(\gamma_i-\gamma^\star)\gamma_i .
\end{split}
\end{equation}
Since \(J\ge p_i\|\vartheta_i\|\) and \(\delta_i\ge0\),
\begin{equation}\label{eq:mixed_delta_bound}
-J\sum_{i=1}^M\delta_i
\le
-\sum_{i=1}^Mp_i\|\vartheta_i\|\delta_i .
\end{equation}
Moreover, by the definition of \(\delta_i\),
\begin{equation}\label{eq:gamma_delta_identity}
\begin{split}
\|\vartheta_i\|(\gamma_i-\delta_i)
&=
\|\vartheta_i\|
\left(
\gamma_i
-
\frac{\gamma_i^2\|\vartheta_i\|}
{\gamma_i\|\vartheta_i\|+\varpi_i(t)}
\right)\\
&=
\frac{\gamma_i\|\vartheta_i\|\varpi_i(t)}
{\gamma_i\|\vartheta_i\|+\varpi_i(t)}
\le
\varpi_i(t).
\end{split}
\end{equation}
Using \eqref{eq:mixed_delta_bound} and \eqref{eq:gamma_delta_identity},
\begin{equation}\label{eq:mixed_term_final}
-J\sum_{i=1}^M\delta_i
+
\sum_{i=1}^Mp_i\gamma_i\|\vartheta_i\|
\le
\sum_{i=1}^Mp_i\varpi_i(t).
\end{equation}
Finally,
\begin{equation}\label{eq:gamma_leak_bound}
-\varpi_i(t)(\gamma_i-\gamma^\star)\gamma_i
\le
\frac{(\gamma^\star)^2}{4}\varpi_i(t),
\end{equation}
because
\[
-(\gamma_i-\gamma^\star)\gamma_i
=
-\left(\gamma_i-\frac{\gamma^\star}{2}\right)^2
+
\frac{(\gamma^\star)^2}{4}
\le
\frac{(\gamma^\star)^2}{4}.
\]
Combining \eqref{eq:Wdot_after_gamma}--\eqref{eq:gamma_leak_bound}, we obtain
\begin{equation}\label{eq:Wdot_net}
D^+W
\le
-c_0J^2
-J
+
C_N\sum_{i=1}^Mp_i\varpi_i(t),
\qquad
C_N=1+\frac{(\gamma^\star)^2}{4}.
\end{equation}
Since \(\varpi_i\in L_1[0,\infty)\), integration of \eqref{eq:Wdot_net} gives
\begin{equation}\label{eq:network_integral_bound}
\begin{split}
&W(t)
+
c_0\int_0^tJ^2(\tau)d\tau
+
\int_0^tJ(\tau)d\tau
\\
&\le
W(0)
+
C_N\sum_{i=1}^Mp_i\int_0^\infty\varpi_i(\tau)d\tau .
\end{split}
\end{equation}
Hence
\[
W\in L_\infty,\qquad
J\in L_1[0,\infty)\cap L_2[0,\infty),\qquad
\gamma_i\in L_\infty[0,\infty).
\]
Since \(J=\sum_i p_i\|\vartheta_i\|\) and \(p_i>0\), it follows that \(\vartheta\in L_\infty\). From \eqref{eq:vartheta_stack},
\begin{equation}\label{eq:sigma_from_vartheta}
\sigma
=
(H_F^{-1}\otimes I_m)\vartheta
-
(H_F^{-1}L_{FL}\otimes I_m)\Phi .
\end{equation}
Assumption~\ref{Assumption: Leaders bound} gives \(\Phi\in L_\infty\), and hence \(\sigma\in L_\infty\). Also,
\[
\|\upsilon_i\|
=
\frac{\gamma_i^2\|\vartheta_i\|}
{\gamma_i\|\vartheta_i\|+\varpi_i(t)}
\le
\gamma_i,
\]
so \(\upsilon\in L_\infty\).

The right-hand side of \eqref{eq:net} is therefore bounded on every finite interval on which the solution exists. Since \(\gamma_i(t)\) remains positive by \eqref{eq:gamma_positive}, and since all components of \((\sigma,\gamma)\) are bounded on finite intervals, the maximal solution cannot escape in finite time. Thus \(T_{\max}=\infty\).

It remains to prove convergence. From \eqref{eq:vartheta_dyn}, and using \(\vartheta,\upsilon,g_\Phi\in L_\infty\), we have \(\dot\vartheta\in L_\infty\) almost everywhere. Hence \(\vartheta\) is uniformly continuous on \([0,\infty)\). Since \(\vartheta\) is bounded and \(J\) is Lipschitz on bounded subsets of \(\mathbb R^{Mm}\), the scalar function \(t\mapsto J(\vartheta(t))\) is uniformly continuous. Because \(J\ge0\) and \(J\in L_1[0,\infty)\), Barbalat's lemma gives
\[
J(t)\to0.
\]
Since \(p_i>0\) for all \(i\), this implies \(\vartheta_i(t)\to0\) for all \(i\). Define
\begin{equation}\label{eq:sigma_star_net}
\sigma^*
=
-(H_F^{-1}L_{FL}\otimes I_m)\Phi .
\end{equation}
Then \eqref{eq:vartheta_stack} gives
\[
\vartheta=(H_F\otimes I_m)(\sigma-\sigma^*).
\]
Since \(H_F\) is nonsingular, \(\vartheta(t)\to0\) implies
\[
\sigma(t)-\sigma^*(t)\to0.
\]
Under Assumption~\ref{Assumption: Graph}, the matrix \(-H_F^{-1}L_{FL}\) is row stochastic. Hence each \(\sigma_i^*(t)\) is a convex combination of \(\phi_1(t),\ldots,\phi_N(t)\), and therefore
\[
\sigma_i^*(t)\in\co(\Phi_L(t)).
\]
Consequently,
\[
\dist(\sigma_i(t),\co(\Phi_L(t)))
\le
\|\sigma_i(t)-\sigma_i^*(t)\|
\to0,
\]
which proves \eqref{eq:sigma_cont}.
\end{proof}

\subsection{Interconnection and Physical-Output Containment}
We now interconnect the network interface with the local recovery layer by setting
\begin{equation}\label{eq:T_sigma_interface}
r_i=T_i\sigma_i,
\qquad
T_i=-(H_iF_i^{-1}G_i)^{-1}.
\end{equation}
The inverse exists by the admissibility condition \(\operatorname{rank}(H_iF_i^{-1}G_i)=m\) in \eqref{eq:embed}. Since Theorem~\ref{thm:network} gives \(\sigma_i\in L_\infty[0,\infty)\), the command \(r_i\) is also bounded and therefore satisfies the input requirement of Theorem~\ref{thm:local}. Define
\begin{equation}\label{eq:Xi_filter_def}
X_i=F_i^{-1}G_i(H_iF_i^{-1}G_i)^{-1}.
\end{equation}
Then, by \eqref{eq:T_sigma_interface},
\begin{equation}\label{eq:filter_interface_identities}
F_iX_i+G_iT_i=0,\qquad H_iX_i=I_m.
\end{equation}
The following result completes the solution of Problem~\ref{prob:main}.

\begin{corollary}\label{cor:main}
Under Assumptions~\ref{Assumption: Graph}--\ref{Assumption: Uncorrelated attack}, consider the interconnected design composed of \eqref{eq:local_ref_gen}, \eqref{eq:T_sigma_interface}, \eqref{eq:unom}, \eqref{eq:va}, and \eqref{eq:net}. Then all closed-loop signals are bounded,
\[
e_{x_i}(t)\to0,\qquad e_{\eta_i}(t)\to0,
\]
and, for every \(i\in\mathcal F\),
\begin{equation}\label{eq:practical}
\limsup_{t\to\infty}
\dist\bigl(y_i(t),\co(\Phi_L(t))\bigr)
\le
\varepsilon_i,
\end{equation}
where
\begin{equation}\label{eq:eps_bound}
\varepsilon_i=
\left(
\int_0^\infty
\|H_ie^{F_is}X_i\|\,ds
\right)
\limsup_{t\to\infty}\|\dot\sigma_i(t)\|.
\end{equation}
\end{corollary}

\begin{proof}
By Theorem~\ref{thm:network}, the network-interface signals are bounded, \(\sigma_i\in C([0,\infty),\mathbb R^m)\cap L_\infty[0,\infty)\), and
\[
\dist(\sigma_i(t),\co(\Phi_L(t)))\to0 .
\]
Hence \(r_i=T_i\sigma_i\in C([0,\infty),\mathbb R^m)\cap L_\infty[0,\infty)\). Applying Theorem~\ref{thm:local} gives bounded local closed-loop signals and
\[
e_{x_i}(t)\to0,\qquad e_{\eta_i}(t)\to0.
\]
Since
\[
y_i=C_ix_i=C_i(e_{x_i}+\Pi_{i_1}z_i)=C_ie_{x_i}+H_iz_i,
\]
the point-to-set distance satisfies
\begin{equation}\label{eq:distance_split_short}
\dist(y_i,\co(\Phi_L))
\le
\|H_iz_i-\sigma_i\|
+
\dist(\sigma_i,\co(\Phi_L))
+
\|C_ie_{x_i}\|.
\end{equation}
The second term on the right-hand side of \eqref{eq:distance_split_short} converges to zero by Theorem~\ref{thm:network}, and the third term converges to zero by Theorem~\ref{thm:local}. It remains to bound the command-filter mismatch.

Set
\[
\tilde z_i=z_i-X_i\sigma_i .
\]
Using \eqref{eq:local_ref_gen}, \eqref{eq:T_sigma_interface}, and \eqref{eq:filter_interface_identities}, one obtains
\begin{equation}\label{eq:realization_error_dyn}
\dot{\tilde z}_i
=
F_i\tilde z_i-X_i\dot\sigma_i,
\qquad
H_i\tilde z_i=H_iz_i-\sigma_i .
\end{equation}
By the variation-of-constants formula,
\[
H_i\tilde z_i(t)
=
H_ie^{F_it}\tilde z_i(0)
-
\int_0^t
H_ie^{F_i(t-\tau)}X_i\dot\sigma_i(\tau)d\tau .
\]
Since \(F_i\) is Hurwitz, the kernel \(K_i(s)=H_ie^{F_is}X_i\) belongs to \(L_1[0,\infty)\), and the homogeneous term \(H_ie^{F_it}\tilde z_i(0)\) converges to zero. For any \(\delta>0\), there exists \(T_\delta>0\) such that
\[
\|\dot\sigma_i(t)\|
\le
\limsup_{\tau\to\infty}\|\dot\sigma_i(\tau)\|+\delta,
\qquad t\ge T_\delta .
\]
In the convolution above, the contribution over \([0,T_\delta]\) converges to zero because the interval is fixed and \(e^{F_i(t-\tau)}\) decays exponentially. The tail over \([T_\delta,t]\) is bounded by
\[
\left(
\limsup_{\tau\to\infty}\|\dot\sigma_i(\tau)\|+\delta
\right)
\int_0^\infty\|H_ie^{F_is}X_i\|\,ds .
\]
Taking the upper limit as \(t\to\infty\), and then letting \(\delta\to0^+\), gives
\[
\begin{split}
&\limsup_{t\to\infty}\|H_iz_i(t)-\sigma_i(t)\|
\le\\
&\left(
\int_0^\infty\|H_ie^{F_is}X_i\|\,ds
\right)
\limsup_{t\to\infty}\|\dot\sigma_i(t)\|.
\end{split}
\]
Combining this bound with \eqref{eq:distance_split_short} proves \eqref{eq:practical}--\eqref{eq:eps_bound}.
\end{proof}

\begin{remark}
The residual \(\varepsilon_i\) quantifies the mismatch between the generated task-space command \(\sigma_i\) and the realizable filtered command \(H_iz_i\). This residual is induced by the stable local command filter and by the fact that the leader trajectories are only assumed to be locally absolutely continuous with unknown velocity bounds. For heterogeneous followers with relative degrees larger than one, exact physical-output tracking of an arbitrary moving command would generally require higher-order command derivatives, which are not available under the information pattern considered here.
\end{remark}

Figure~\ref{fig:two_layer_architecture} illustrates an overview of the proposed two-layer resilient output-containment architecture for each follower \(i\in\mathcal F\). The first layer is the local execution and recovery layer. It converts the network command \(\sigma_i\) into a realizable local command \(r_i=T_i\sigma_i\), processes it through the admissible stable command filter \(\dot z_i=F_iz_i+G_ir_i\), and applies the nominal embedding controller \(u_{i,\mathrm{nom}}=\psi_i^{-1}(Q_iz_i+K_ir_i+u_i^r)\). The adaptive virtual-actuator recovery term \(u_i^r\) is then generated from the local tracking variable \(s_i=B_i^TP_ie_{x_i}\) to compensate for state-correlated, input-correlated, and bounded exogenous actuator attacks. Importantly, the local recovery layer uses only the available external normal-form state \(x_i\), and it does not require measurement of the zero-dynamics state \(\eta_i\).

The second layer is the network-interface layer, which generates the task-space command \(\sigma_i\) using only neighbor-exchanged interface states \(\sigma_j\) and instantaneous leader-neighbor task-space signals \(\phi_\ell\), when such leader-neighbor links are present. This layer does not require knowledge of the leaders' dynamics, leader velocity bounds, leader motion envelopes, or global graph parameters to tune the controller. Its role is to enforce command-level containment by driving the generated command \(\sigma_i\) to the convex hull of the leaders' task-space signals.

This two-layer separation also prevents actuator-side compromises from entering the network-interface dynamics, while the local layer ensures resilient realization of the generated command at the physical follower level.
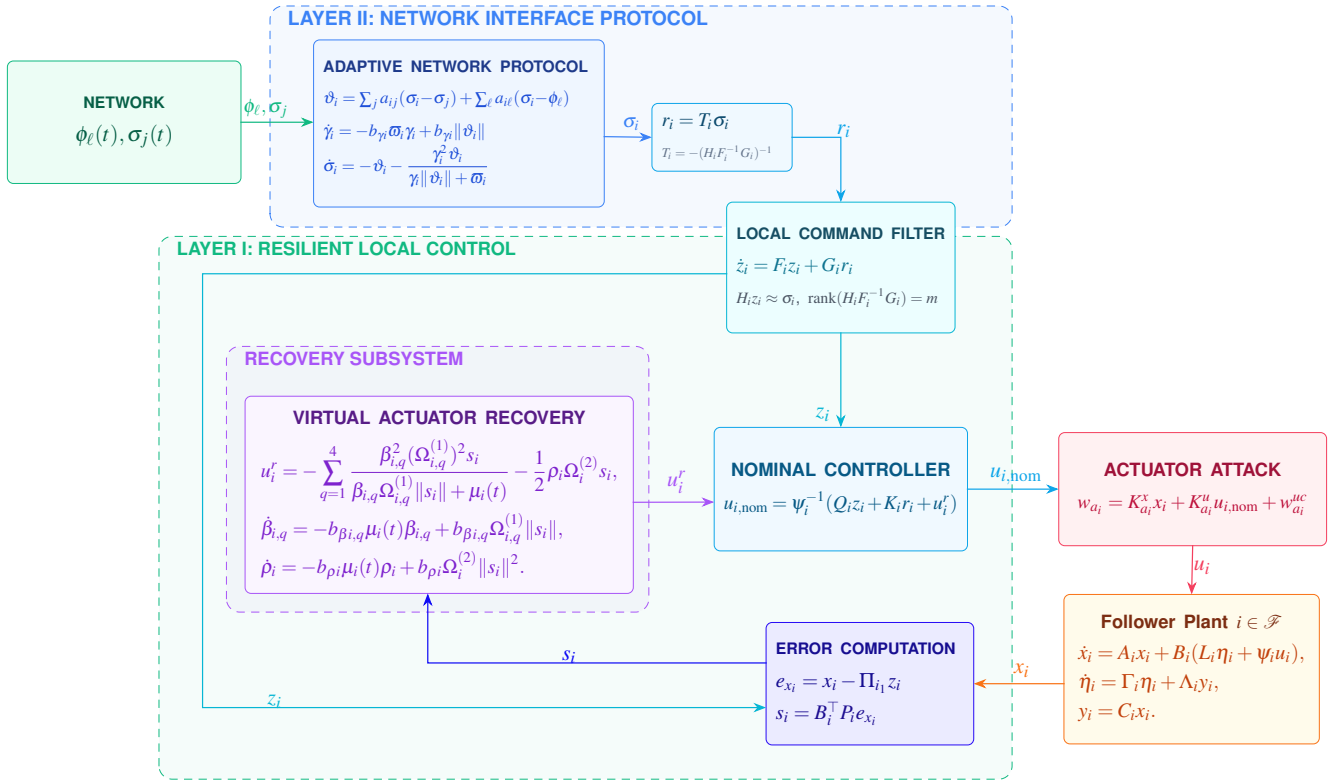
\begin{figure*}[!t]
\centering
\begin{adjustbox}{max totalsize={\textwidth}{\textheight},center}
\begin{tikzpicture}[font=\sffamily,>=Stealth,line cap=rect]
\path[use as bounding box] (0.285,0.825) rectangle (19.710,12.360);

\fill[C3B82F6,fill opacity=0.050,rounded corners=5.5pt] (4.260,8.985) rectangle (14.985,12.180);
\draw[C3B82F6,dashed,line width=0.35pt,rounded corners=5.5pt] (4.260,8.985) rectangle (14.985,12.180);
\node[C3B82F6,anchor=west,font=\fontsize{8pt}{8pt}\selectfont\bfseries\sffamily] at (4.395,12.000) {\MakeUppercase{Layer II: Network Interface Protocol}};

\fill[C10B981,fill opacity=0.040,rounded corners=5.5pt] (2.655,1.005) rectangle (14.985,8.850);
\draw[C10B981,dashed,line width=0.35pt,rounded corners=5.5pt] (2.655,1.005) rectangle (14.985,8.850);
\node[C10B981,anchor=west,font=\fontsize{7.5pt}{8pt}\selectfont\bfseries\sffamily] at (2.790,8.670) {\MakeUppercase{Layer I: Resilient Local Control}};

\fill[CA855F7,fill opacity=0.040,rounded corners=5.5pt] (3.630,3.420) rectangle (9.735,7.260);
\draw[CA855F7,dashed,line width=0.35pt,rounded corners=5.5pt] (3.630,3.420) rectangle (9.735,7.260);
\node[CA855F7,anchor=west,font=\fontsize{7.5pt}{8pt}\selectfont\bfseries\sffamily] at (3.765,7.080) {\MakeUppercase{Recovery Subsystem}};

\node[
draw=C10B981,
fill=CECFDF5,
rounded corners=2.5pt,
text width=3.095cm,
minimum width=3.375cm,
minimum height=1.770cm,
align=center,
inner sep=4.0pt
] (L1) at (2.152,10.500) {
{\fontsize{6.5pt}{7.5pt}\selectfont\bfseries\color{C065F46}\MakeUppercase{Network}}\\[1.5pt]
{\fontsize{9pt}{9pt}\selectfont\color{C065F46}\(\phi_\ell(t),\sigma_j(t)\)}\\[1pt]
};

\node[
draw=C3B82F6,
fill=CEFF6FF,
rounded corners=2.5pt,
text width=3.920cm,
minimum width=4.200cm,
minimum height=2.415cm,
align=left,
inner sep=4.0pt
] (PR) at (6.990,10.482) {
{\fontsize{6.5pt}{7.5pt}\selectfont\bfseries\color{C1E3A8A}\MakeUppercase{Adaptive Network Protocol}}\\[1pt]
{\fontsize{7pt}{9pt}\selectfont\color{C1D4ED8}\(\vartheta_i=\sum_j a_{ij}(\sigma_i{-}\sigma_j)+\sum_\ell a_{i\ell}(\sigma_i{-}\phi_\ell)\)}\\[1.5pt]
{\fontsize{7pt}{9pt}\selectfont\color{C1D4ED8}\(\dot\gamma_i=-b_{\gamma i}\varpi_i\gamma_i+b_{\gamma i}\|\vartheta_i\|\)}\\[1.5pt]
{\fontsize{7pt}{9pt}\selectfont\color{C1D4ED8}\(\dot\sigma_i=-\vartheta_i-\dfrac{\gamma_i^2\vartheta_i}{\gamma_i\|\vartheta_i\|+\varpi_i}\)}
};

\node[
draw=C0EA5E9,
fill=CF0F9FF,
rounded corners=2.5pt,
text width=1.745cm,
minimum width=2.025cm,
minimum height=0.975cm,
align=left,
inner sep=4.0pt
] (TI) at (10.793,10.282) {
{\fontsize{8pt}{9pt}\selectfont\color{C0C4A6E}\(r_i=T_i\sigma_i\)}\\[1pt]
{\fontsize{5.5pt}{6.5pt}\selectfont\color{C64748B}\(T_i=-(H_iF_i^{-1}G_i)^{-1}\)}
};

\node[
draw=C06B6D4,
fill=CECFEFF,
rounded corners=2.5pt,
text width=3.020cm,
minimum width=3.300cm,
minimum height=1.890cm,
align=left,
inner sep=4.0pt
] (CF) at (12.510,8.400) {
{\fontsize{6.5pt}{7.5pt}\selectfont\bfseries\color{C164E63}\MakeUppercase{Local Command Filter}}\\[1pt]
{\fontsize{8pt}{9pt}\selectfont\color{C0E7490}\(\dot z_i=F_iz_i+G_ir_i\)}\\[1pt]
{\fontsize{6.5pt}{7.5pt}\selectfont\color{C4B5563}\(H_iz_i\approx\sigma_i\), \(\operatorname{rank}(H_iF_i^{-1}G_i)=m\)}
};

\node[
draw=C0EA5E9,
fill=CF0F9FF,
rounded corners=2.5pt,
text width=3.380cm,
minimum width=3.660cm,
minimum height=1.785cm,
align=center,
inner sep=4.0pt
] (NC) at (12.510,5.197) {
{\fontsize{7.5pt}{7.5pt}\selectfont\bfseries\color{C075985}\MakeUppercase{Nominal Controller}}\\[1.5pt]
{\fontsize{8pt}{9pt}\selectfont\color{C075985}\(u_{i,\mathrm{nom}}=\psi_i^{-1}(Q_iz_i+K_ir_i+u_i^r)\)}\\[1pt]
};

\node[
draw=CA855F7,
fill=CFAF5FF,
rounded corners=2.5pt,
text width=5.330cm,
minimum width=5.610cm,
minimum height=2.850cm,
align=center,
inner sep=4.0pt
] (VA) at (6.705,5.115) {
{\fontsize{7.5pt}{7.5pt}\selectfont\bfseries\color{C581C87}\MakeUppercase{Virtual Actuator Recovery}}\\[6pt]
{\fontsize{8.2pt}{7.4pt}\selectfont\color{C7E22CE}
\(\begin{aligned}
&u_i^r=-\sum_{q=1}^{4}\frac{\beta_{i,q}^{2}(\Omega_{i,q}^{(1)})^{2}s_i}{\beta_{i,q}\Omega_{i,q}^{(1)}\|s_i\|+\mu_i(t)}-\frac12\rho_i\Omega_i^{(2)}s_i,\\
&\dot\beta_{i,q}=-b_{\beta i,q}\mu_i(t)\beta_{i,q}+b_{\beta i,q}\Omega_{i,q}^{(1)}\|s_i\|,\\
&\dot\rho_i=-b_{\rho i}\mu_i(t)\rho_i+b_{\rho i}\Omega_i^{(2)}\|s_i\|^2 .
\end{aligned}\)}
};

\node[
draw=C260AF5,
fill=CEBEBFF,
rounded corners=2.5pt,
text width=2.720cm,
minimum width=3.000cm,
minimum height=1.770cm,
align=left,
inner sep=4.0pt
] (EC) at (12.930,2.385) {
{\fontsize{6.5pt}{7.5pt}\selectfont\bfseries\color{C160D91}\MakeUppercase{Error Computation}}\\[1.5pt]
{\fontsize{8.5pt}{9pt}\selectfont\color{C160D91}\(e_{x_i}=x_i-\Pi_{i_1}z_i\)}\\[1.5pt]
{\fontsize{8.5pt}{9pt}\selectfont\color{C1D0D91}\(s_i=B_i^\top P_ie_{x_i}\)}\\[1pt]
};

\node[
draw=CF43F5E,
fill=CFFF1F2,
rounded corners=2.5pt,
text width=3.590cm,
minimum width=3.870cm,
minimum height=1.635cm,
align=center,
inner sep=4.0pt
] (AT) at (17.595,5.197) {
{\fontsize{7.5pt}{7.5pt}\selectfont\bfseries\color{C9F1239}\MakeUppercase{Actuator Attack}}\\[1pt]
{\fontsize{8pt}{9pt}\selectfont\color{CBE123C}\(w_{a_i}=K_{a_i}^xx_i+K_{a_i}^u u_{i,\mathrm{nom}}+w_{a_i}^{uc}\)}\\[1pt]
};

\node[
draw=CF97316,
fill=CFFFBEB,
rounded corners=2.5pt,
text width=3.410cm,
minimum width=3.690cm,
minimum height=2.160cm,
align=center,
inner sep=4.0pt
] (PL) at (17.580,2.595) {
{\fontsize{7.5pt}{7.5pt}\selectfont\bfseries\color{C7C2D12}Follower Plant \(i\in\mathcal F\)}\\[6pt]
{\fontsize{8.5pt}{9pt}\selectfont\color{CC2410C}
\(\begin{aligned}
&\dot x_i=A_ix_i+B_i(L_i\eta_i+\psi_iu_i),\\
&\dot\eta_i=\Gamma_i\eta_i+\Lambda_i y_i,\\
&y_i=C_ix_i .
\end{aligned}\)}
};

\draw[->,C10B981,line width=0.5pt] (3.840,10.500) -- (4.890,10.500);
\node[C10B981,font=\fontsize{8pt}{8pt}\selectfont\sffamily] at (4.20,10.70) {\(\phi_\ell,\sigma_j\)};

\draw[->,C3B82F6,line width=0.5pt] (9.090,10.292) -- (9.780,10.282);
\node[C3B82F6,font=\fontsize{8pt}{8pt}\selectfont\sffamily] at (9.480,10.473) {\(\sigma_i\)};

\draw[->,C0EA5E9,line width=0.5pt] (11.805,10.282) -- (12.510,10.282) -- (12.510,9.345);
\node[C0EA5E9,font=\fontsize{9pt}{9pt}\selectfont\sffamily] at (12.555,10.373) {\(r_i\)};

\draw[->,C06B6D4,line width=0.5pt] (12.510,7.455) -- (12.510,6.090);
\node[C06B6D4,font=\fontsize{9pt}{9pt}\selectfont\sffamily] at (12.255,6.280) {\(z_i\)};

\draw[->,C06B6D4,line width=0.5pt] (10.860,8.310) -- (3.285,8.310) -- (3.285,2.040) -- (11.430,2.040);
\node[C06B6D4,font=\fontsize{9pt}{9pt}\selectfont\sffamily] at (4.330,2.130) {\(z_i\)};

\draw[->,C0E0AF5,line width=0.5pt] (11.430,2.685) -- (6.540,2.685) -- (6.540,3.690);
\node[C0E0AF5,font=\fontsize{9pt}{9pt}\selectfont\sffamily] at (8.585,2.775) {\(s_i\)};

\draw[->,CA855F7,line width=0.5pt] (9.510,5.007) -- (10.680,5.010);
\node[CA855F7,font=\fontsize{9pt}{9pt}\selectfont\sffamily] at (10.140,5.300) {\(u_i^r\)};

\draw[->,C0EA5E9,line width=0.5pt] (14.340,5.197) -- (15.660,5.197);
\node[C0EA5E9,font=\fontsize{9pt}{9pt}\selectfont\sffamily] at (15.045,5.387) {\(u_{i,\mathrm{nom}}\)};

\draw[->,CF43F5E,line width=0.5pt] (17.595,4.380) -- (17.580,3.675);
\node[CF43F5E,font=\fontsize{9pt}{9pt}\selectfont\sffamily] at (17.735,4.070) {\(u_i\)};

\draw[->,CF97316,line width=0.5pt] (15.735,2.385) -- (14.430,2.385);
\node[CF97316,font=\fontsize{9pt}{9pt}\selectfont\sffamily] at (15.127,2.575) {\(x_i\)};

\end{tikzpicture}
\end{adjustbox}
\caption{Two-layer resilient output-containment architecture for follower \(i\). The network-interface layer generates the task-space command \(\sigma_i\) using only neighbor-exchanged interface states and, if available, instantaneous leaders' signals, without requiring leader dynamics, leader velocity bounds, leader motion envelopes, or global graph parameters for tuning. The local layer maps \(\sigma_i\) into \(r_i=T_i\sigma_i\), realizes it through an admissible stable command filter, and applies a nominal embedding controller augmented by adaptive virtual-actuator recovery. The recovery term is computed adaptively from the available external normal-form state \(x_i\) through \(s_i=B_i^TP_ie_{x_i}\), and compensates for state-correlated, input-correlated, and bounded exogenous actuator attacks without requiring measurement of the zero-dynamics state \(\eta_i\).}
\label{fig:two_layer_architecture}
\vspace{-1mm}
\end{figure*}
\section{Numerical Simulation}
\label{sec:simulation}
We validate the proposed architecture on a heterogeneous output-containment problem with six followers and three leaders. Each follower represents the lateral motion of a quadrotor carrying a cable-suspended payload, with controlled output
\[
y_i=\begin{bmatrix}p_{x_i}&p_{y_i}\end{bmatrix}^{T}\in\mathbb R^2 .
\]
The followers have heterogeneous physical parameters, and their outputs have vector relative degree \([4,4]^T\). The leaders generate a moving, rotating, and resizing convex hull whose dynamical model and velocity bounds are not available to the followers. Two of the followers are also subjected to persistent actuator attacks combining state-correlated, input-correlated, and bounded exogenous false-data components. The simulation therefore verifies command-level containment of the generated interface states \(\sigma_i\), resilient local realization of the admissible command-filter trajectories, and practical physical-output containment with the residual characterized in Corollary~\ref{cor:main}.

\subsection{Heterogeneous Follower Model}
\label{subsec:sim_model}
The follower model is obtained from a small-angle near-hover linearization of a quadrotor with a cable-suspended payload. The yaw angle is fixed over the lateral-control time scale, and the roll and pitch angles as well as the load swing angles are small. For follower \(i\), let \(m_{q_i}\) denote the quadrotor mass, \(m_{\ell_i}\) the payload mass, \(\ell_{a_i}\) the arm length, \(J_{x_i}\) and \(J_{y_i}\) the roll and pitch inertias, and \(\ell_{x_i}\), \(\ell_{y_i}\) the effective cable lengths in the two lateral planes. The loaded hover thrust is taken as
\[
T_{0_i}=(m_{q_i}+m_{\ell_i})g_0,\qquad g_0=9.807\,\mathrm{m/s^2}.
\]
For the \(x\)-axis, the linearized translational-pendulum equations are
\[
(m_{q_i}+m_{\ell_i})\ddot p_{x_i}+m_{\ell_i}\ell_{x_i}\ddot\alpha_{x_i}=T_{0_i}\theta_i,
\]
\[
\ddot p_{x_i}+\ell_{x_i}\ddot\alpha_{x_i}+2\zeta_{x_i}w_{x_i}\ell_{x_i}\dot\alpha_{x_i}+g_0\alpha_{x_i}=0,
\qquad
w_{x_i}=\sqrt{g_0/\ell_{x_i}} .
\]
Eliminating \(\ddot\alpha_{x_i}\) gives
\[
\ddot p_{x_i}=a_{\theta i}\theta_i+a_{\alpha x,i}\alpha_{x_i}+a_{\beta x,i}\dot\alpha_{x_i},
\]
where
\[
a_{\theta i}=\frac{T_{0_i}}{m_{q_i}},\qquad
a_{\alpha x,i}=\frac{m_{\ell_i}g_0}{m_{q_i}},\qquad
a_{\beta x,i}=\frac{m_{\ell_i}\ell_{x_i}(2\zeta_{x_i}w_{x_i})}{m_{q_i}}.
\]
Similarly, after absorbing the roll-channel sign into the \(y\)-axis input convention,
\[
\ddot p_{y_i}=a_{\phi i}\phi_i+a_{\alpha y,i}\alpha_{y_i}+a_{\beta y,i}\dot\alpha_{y_i},
\]
with \(w_{y_i}=\sqrt{g_0/\ell_{y_i}}\), we have
\[
a_{\phi i}=\frac{T_{0_i}}{m_{q_i}},\qquad
a_{\alpha y,i}=\frac{m_{\ell_i}g_0}{m_{q_i}},\qquad
a_{\beta y,i}=\frac{m_{\ell_i}\ell_{y_i}(2\zeta_{y_i}w_{y_i})}{m_{q_i}}.
\]
The attitude channels are modeled by
\[
\ddot\theta_i=b_{\theta i}u_{x_i},\qquad
\ddot\phi_i=b_{\phi i}u_{y_i},
\qquad
b_{\theta i}=\frac{\ell_{a_i}}{J_{y_i}},\qquad
b_{\phi i}=\frac{\ell_{a_i}}{J_{x_i}}.
\]
Thus the physical high-frequency matrix from the actuator input to the fourth output derivative is
\[
\psi_i^0=
\begin{bmatrix}
a_{\theta i}b_{\theta i}&0\\
0&a_{\phi i}b_{\phi i}
\end{bmatrix}
=
\begin{bmatrix}
\dfrac{T_{0_i}\ell_{a_i}}{m_{q_i}J_{y_i}}&0\\[1mm]
0&\dfrac{T_{0_i}\ell_{a_i}}{m_{q_i}J_{x_i}}
\end{bmatrix},
\]
which is nonsingular. Hence the physical input-output model has vector relative degree \([4,4]^T\).

The raw physical state is chosen as
\[
X_i^0=\operatorname{col}
\bigl(
p_{x_i},\dot p_{x_i},\theta_i,\dot\theta_i,\alpha_{x_i},\dot\alpha_{x_i},
p_{y_i},\dot p_{y_i},\phi_i,\dot\phi_i,\alpha_{y_i},\dot\alpha_{y_i}
\bigr).
\]
It satisfies
\[
\dot X_i^0=A_i^0X_i^0+B_i^0u_i,\qquad y_i=C_i^0X_i^0,
\]
where \(u_i=\operatorname{col}(u_{x_i},u_{y_i})\), \(C_i^0\) selects \(p_{x_i}\) and \(p_{y_i}\), and the nonzero entries of \(A_i^0\) and \(B_i^0\) are
\[
\begin{split}
&A_i^0(1,2)=1, A_i^0(2,3)=a_{\theta i}, A_i^0(2,5)=a_{\alpha x,i},\\
&A_i^0(2,6)=a_{\beta x,i},A_i^0(3,4)=1,B_i^0(4,1)=b_{\theta i},A_i^0(5,6)=1,\\
&A_i^0(6,3)=-a_{\theta i}/\ell_{x_i},A_i^0(6,5)=-(w_{x_i}^2+a_{\alpha x,i}/\ell_{x_i}),\\
&A_i^0(6,6)=-(2\zeta_{x_i}w_{x_i}+a_{\beta x,i}/\ell_{x_i}),
\end{split}
\]
and
\[
\begin{split}
&A_i^0(7,8)=1,A_i^0(8,9)=a_{\phi i},A_i^0(8,11)=a_{\alpha y,i},\\
&A_i^0(8,12)=a_{\beta y,i},A_i^0(9,10)=1,B_i^0(10,2)=b_{\phi i},\\
&A_i^0(11,12)=1,A_i^0(12,9)=-a_{\phi i}/\ell_{y_i},\\
&A_i^0(12,11)=-(w_{y_i}^2+a_{\alpha y,i}/\ell_{y_i}),\\
&A_i^0(12,12)=-(2\zeta_{y_i}w_{y_i}+a_{\beta y,i}/\ell_{y_i}).
\end{split}
\]
All other entries are zero.

The physical parameters are heterogeneous across the six followers. The quadrotor masses satisfy \(m_{q_i}\in[1.44,1.56]\,\mathrm{kg}\), the payload masses satisfy \(m_{\ell_i}\in[0.301,0.346]\,\mathrm{kg}\), and the arm lengths satisfy \(\ell_{a_i}\in[0.218,0.232]\,\mathrm{m}\). The inertias are \(J_{x_i}\in[2.068,2.332]\times10^{-2}\,\mathrm{kg\,m^2}\) and \(J_{y_i}\in[2.090,2.332]\times10^{-2}\,\mathrm{kg\,m^2}\). The cable lengths are \(\ell_{x_i}\in[0.576,0.636]\,\mathrm{m}\) and \(\ell_{y_i}\in[0.523,0.578]\,\mathrm{m}\), and the damping ratios are \(\zeta_{x_i}\in[0.523,0.594]\) and \(\zeta_{y_i}\in[0.528,0.583]\).

The physical model is then transformed into the normal form \eqref{eq: follower_dynamics_lcss} using the algorithm of \cite{mueller2009normal}, leading to
\[
\bar A_i=
\begin{bmatrix}
A_i&B_iL_i\\
\Lambda_iC_i&\Gamma_i
\end{bmatrix},
\qquad
\bar B_i=
\begin{bmatrix}
B_i\\
0
\end{bmatrix}.
\]
These matrices are the same block matrices that appear in \eqref{eq:embed}. The transformation was verified numerically, and the largest real part among the internal-dynamics eigenvalues was \(-2.156\). Hence, all transformed internal dynamics are Hurwitz, consistent with Assumption~\ref{Assumption: Minimum-phase}.

\subsection{Command-Filter Construction}
\label{subsec:sim_command_filter}
The command filter is constructed after the physical model has been transformed into the normal form \eqref{eq: follower_dynamics_lcss}. Since each follower has vector relative degree \([4,4]^T\), the external state dimension is \(8\), the internal dimension is \(4\), and the local filter state is selected as \(z_i\in\mathbb R^{12}\).

The external part of the filter assigns stable fourth-order chains in the two output channels. Let
\[
A_{\mathrm c}=
\begin{bmatrix}
0&1&0&0\\
0&0&1&0\\
0&0&0&1\\
0&0&0&0
\end{bmatrix},
\qquad
b_{\mathrm c}=
\begin{bmatrix}
0\\0\\0\\1
\end{bmatrix}.
\]
For the command filters, we use a common bandwidth \(\omega_{\mathrm{cf}}=10.5\,\mathrm{rad/s}\) for all followers, together with a follower-dependent scaling factor to induce heterogeneity at the local controller level:
\[
\kappa_i^{\mathrm{cf}}=0.94+0.025(i-1),
\qquad i=1,\ldots,6 .
\]
The assigned \(x\)- and \(y\)-channel pole sets are therefore
\[
\lambda_{x_i}^{\mathrm{cf}}
=
-\omega_{\mathrm{cf}}\kappa_i^{\mathrm{cf}}
\begin{bmatrix}
0.45&0.55&0.65&0.75
\end{bmatrix},
\]
\[
\lambda_{y_i}^{\mathrm{cf}}
=
-\omega_{\mathrm{cf}}\kappa_i^{\mathrm{cf}}
\begin{bmatrix}
0.47&0.57&0.67&0.73
\end{bmatrix}.
\]
Thus, across the six followers, the assigned \(x\)-chain poles lie in \([-8.39,-4.44]\), and the assigned \(y\)-chain poles lie in \([-8.16,-4.64]\).

Let \(k_{x_i}^{\mathrm{cf}}\in\mathbb R^{1\times4}\) and \(k_{y_i}^{\mathrm{cf}}\in\mathbb R^{1\times4}\) be the pole-placement rows satisfying
\[
\operatorname{spec}\bigl(A_{\mathrm c}-b_{\mathrm c}k_{x_i}^{\mathrm{cf}}\bigr)
=
\lambda_{x_i}^{\mathrm{cf}},
\qquad
\operatorname{spec}\bigl(A_{\mathrm c}-b_{\mathrm c}k_{y_i}^{\mathrm{cf}}\bigr)
=
\lambda_{y_i}^{\mathrm{cf}} .
\]
With the normal-form input-selector convention used for \(B_i\), the external-chain feedback is
\[
Q_i^{\mathrm{cf}}
=
-
\begin{bmatrix}
k_{x_i}^{\mathrm{cf}}&0_{1\times4}\\
0_{1\times4}&k_{y_i}^{\mathrm{cf}}
\end{bmatrix}.
\]
Based on this choice, we select
\[
Q_i=
\begin{bmatrix}
Q_i^{\mathrm{cf}}&-L_i
\end{bmatrix},
\qquad
F_i=\bar A_i+\bar B_iQ_i .
\]
Equivalently,
\[
F_i=
\begin{bmatrix}
A_i+B_iQ_i^{\mathrm{cf}}&0_{8\times4}\\
\Lambda_iC_i&\Gamma_i
\end{bmatrix}.
\]
Therefore, \(F_i\) is block lower triangular. Its external block has the assigned stable chain poles, and its internal block is \(\Gamma_i\), which is Hurwitz by the verified minimum-phase property of the transformed suspended-load model. Hence, \(F_i\) is Hurwitz for every simulated follower.

The immersion matrices are selected as
\[
\Pi_{i_1}=
\begin{bmatrix}
I_8&0_{8\times4}
\end{bmatrix},
\qquad
\Pi_{i_2}=
\begin{bmatrix}
0_{4\times8}&I_4
\end{bmatrix},
\qquad
H_i=C_i\Pi_{i_1}.
\]
Thus,
\[
\Pi_i=\operatorname{col}(\Pi_{i_1},\Pi_{i_2})=I_{12}.
\]
With the above construction,
\[
\Pi_iF_i-\bar A_i\Pi_i-\bar B_iQ_i=0,
\qquad
\Pi_iG_i-\bar B_iK_i=0
\]
once \(G_i\) is chosen as below. The second identity also gives \(\Pi_{i_2}G_i=0\), which is the implemented full-order admissibility condition for the internal part of the filter.

The input matrix \(G_i\) is selected from the DC map between the filter input and the reconstructed task-space output:
\[
D_i^{\mathrm{cf}}=H_iF_i^{-1}\bar B_i .
\]
The assigned fourth-order chain poles are all nonzero, so this DC map is nonsingular. The simulation code verifies this rank condition for each follower before accepting the model. The command-filter input gain is then
\[
K_i=-(D_i^{\mathrm{cf}})^{-1},
\qquad
G_i=\bar B_iK_i .
\]
Consequently,
\[
H_iF_i^{-1}G_i
=
D_i^{\mathrm{cf}}K_i
=
-I_2 .
\]
The command-interface matrix in \eqref{eq:T_sigma_interface} is therefore
\[
T_i=-(H_iF_i^{-1}G_i)^{-1}=I_2 .
\]
Thus, in the reported implementation, \(r_i=\sigma_i\); nevertheless, the construction keeps \(T_i\) explicitly in the notation because it is part of the general interconnection formula used in Corollary~\ref{cor:main}.

Finally, the matrix used in the filter-residual characterization is
\[
X_i=
F_i^{-1}G_i
\bigl(H_iF_i^{-1}G_i\bigr)^{-1}.
\]
With the above choice of \(T_i\), this matrix satisfies
\[
F_iX_i+G_iT_i=0,
\qquad
H_iX_i=I_2,
\]
which are precisely the identities used to connect the network-interface command \(\sigma_i\) to the stable local command filter.

\subsection{Actuator-Attack Scenario}
\label{subsec:sim_attacks}
In the simulation, only followers \(1\) and \(6\) are subjected to attack. Follower \(1\) is subjected to a combined but moderate state- and input-correlated attack plus an exogenous false-data injection, and follower \(6\) is subjected to a hijacking-type actuator compromise.

For follower \(1\), the injected actuator signal is
\[
w_{a_1}
=
K_{a_1}^{x}x_1+E_1u_{1,\mathrm{nom}}+d_{a_1}(t),
\]
where the input-correlated attack gain \(E_1\equiv K^{u}_{a_1}\) is
\[
E_1
=\begin{bmatrix}
-0.490&0.056\\
-0.042&-0.448
\end{bmatrix}.
\]
The bounded exogenous false-data component is
\[
d_{a_1}(t)=
\begin{bmatrix}
2.00\sin(0.031t+0.40)\\
1.50\cos(0.027t+0.70)
\end{bmatrix}.
\]
The state-correlated component is defined in the normal-form external state \(x_1\in\mathbb R^8\), not in the raw physical state. Let
\[
S_x=
\operatorname{diag}(1200,300,120,60,1200,300,120,60),
\]
and define
\[
\bar k_{1x}=
\begin{bmatrix}
1.00&0.35&0.12&0.04&0.05&0.02&0&0
\end{bmatrix},
\]
\[
\bar k_{1y}=
\begin{bmatrix}
-0.04&0&0.02&0&0.90&0.32&0.11&0.04
\end{bmatrix}.
\]
The shaped state-correlated gain is
\[
K_{\mathrm{sh},1}
=
\begin{bmatrix}
\bar k_{1x}\\
\bar k_{1y}
\end{bmatrix}
S_x^{-1},
\qquad
K_{a_1}^{x}=\psi_1^{-1}K_{\mathrm{sh},1}.
\]
The corresponding effective input matrix
\[
\Delta_{u_1}
=
I_2+\psi_1E_1\psi_1^{-1}
\]
satisfies Assumption~\ref{Assumption: Input Correlated Attack}. The exogenous component \(d_{a_1}\) satisfies Assumption~\ref{Assumption: Uncorrelated attack}, and since \(K_{a_1}^{x}\) is constant, the state-correlated term satisfies Assumption~\ref{Assumption: state-correlated attack envelope} with a constant envelope.

For the follower-6 attack, the input-correlated component is
\[
E_6=-0.95I_2 .
\]
Consequently,
\[
\Delta_{u_6}
=
I_2+\psi_6E_6\psi_6^{-1}
=
0.05I_2,
\]
so the direct input authority is reduced to \(5\%\), but remains strictly positive. Thus Assumption~\ref{Assumption: Input Correlated Attack} holds with \(\chi_6=0.05\).

The attack is designed to force the quadrotor to track the output of
\[
\dot\zeta_{a6}=S_{a6}\zeta_{a6},
\qquad
S_{a6}=
\begin{bmatrix}
0&0.30&0&0\\
-0.30&0&0&0\\
0&0&0&0.24\\
0&0&-0.24&0
\end{bmatrix},
\]
with selector
\[
C_{a6}=
\begin{bmatrix}
1&0&0&0\\
0&0&1&0
\end{bmatrix}.
\]
The two oscillator frequencies are \(0.30\,\mathrm{rad/s}\) and \(0.24\,\mathrm{rad/s}\), corresponding to periods of approximately \(20.9\,\mathrm{s}\) and \(26.2\,\mathrm{s}\). The initial amplitudes are \(A_{6x}=360\) and \(A_{6y}=260\), with phases \(0\) and \(\pi\), respectively. Hence, the injected exogenous component remains persistently time-varying over the simulation interval. Define
\[
d_k(t)=C_{a6}S_{a6}^{k}\zeta_{a6}(t),
\qquad
k=0,1,2,3,4.
\]
The corresponding adversarial external-chain trajectory is
\[
x_{a6}(t)=
\operatorname{col}
\bigl(
d_{0,1},d_{1,1},d_{2,1},d_{3,1},
d_{0,2},d_{1,2},d_{2,2},d_{3,2}
\bigr).
\]
The internal target associated with this adversarial chain is obtained from the Sylvester equation
\[
\Pi_{a6}S_{a6}
=
\Gamma_6\Pi_{a6}
+
\Lambda_6C_{a6},
\]
and is given by
\[
\eta_{a6}^{\star}(t)=\Pi_{a6}\zeta_{a6}(t).
\]
This construction gives
\[
\dot\eta_{a6}^{\star}
=
\Gamma_6\eta_{a6}^{\star}
+
\Lambda_6C_6x_{a6}.
\]

The hijacking feedback matrix \(K_{\mathrm{tr},6}\) is then selected in the external normal-form coordinates. Let
\[
\bar p_{h6}
=
\begin{bmatrix}
0.80&1.10&1.50&2.00&0.76&1.05&1.42&1.90
\end{bmatrix}.
\]
The initialization routine searches over \(\mathcal S_h=[0.2,20]\) and chooses the first \(\lambda_{h6}\in\mathcal S_h\) for which the feedback matrix \(K_{\mathrm{tr},6}\) satisfying
\[
\operatorname{spec}(A_6-B_6K_{\mathrm{tr},6})
=
-\lambda_{h6}\bar p_{h6}
\]
also makes the coupled hijacking-error matrix
\[
A_{h6}=
\begin{bmatrix}
A_6-B_6K_{\mathrm{tr},6}&B_6L_6\\
\Lambda_6C_6&\Gamma_6
\end{bmatrix}
\]
Hurwitz, with \(\max_{\lambda\in\operatorname{spec}(A_{h6})}\operatorname{Re}\lambda<-2.0\times10^{-2}\). The state-correlated hijacking gain is then chosen as
\[
K_{a_6}^{x}
=-\psi_6^{-1}K_{\mathrm{tr},6}.
\]
The bounded exogenous term is
\[
d_{a_6}(t)
=
\psi_6^{-1}
\left(
d_4(t)
-
L_6\eta_{a6}^{\star}(t)
+
K_{\mathrm{tr},6}x_{a6}(t)
\right).
\]
Therefore the actuator attack applied to follower \(6\) is
\[
w_{a_6}
=
K_{a_6}^{x}x_6
+
E_6u_{6,\mathrm{nom}}
+
d_{a_6}(t).
\]

\subsection{Controller Parameters}
\label{subsec:sim_controller_parameters}
The recovery controller is implemented in the transformed normal-form coordinates. The Riccati matrices \(P_i\) solving \eqref{eq:are} are computed from the follower-specific normal-form data with search initialized at
\[
\begin{split}
&\alpha_i=8000,\\
&\Sigma_i=50\operatorname{diag}(\bar\sigma,\bar\sigma),\\
&\bar\sigma=(1.20,0.70,0.35,0.12).
\end{split}
\]
The resulting certificates in Lemma~\ref{lem:normalformdesign} are strictly positive for all followers, with \(c_i=\lambda_{\min}(\Xi_i)\) ranging over \([3.51,5.19]\times10^{-2}\) and \(\min_i c_i=3.51\times10^{-2}\). Therefore, the local dissipation inequality \eqref{eq:Vdotlocal} is certified for every follower. The virtual actuator then uses the embedding errors
\[
e_{x_i}=x_i-\Pi_{i_1}z_i,
\]
to generate \(s_i=B_i^TP_ie_{x_i}\). We set
\[
\mu_i(t)=\mu_{0i}\exp(-\lambda_{\mu i}t),
\qquad
\mu_{0i}=3.0,
\qquad
\lambda_{\mu i}=0.002 .
\]
The adaptation gains are \(b_{\beta i,q}=0.08\), \(q=1,\ldots,4\), and \(b_{\rho i}=0.03\), and the adaptive variables are initialized at \(\beta_{i,q}(0)=\rho_i(0)=0.5>0\), satisfying the positivity condition used in Theorem~\ref{thm:local}.

The network interface \eqref{eq:net} uses \(\gamma_i(0)=100\), \(b_{\gamma i}=0.75\), and
\[
\varpi_i(t)=
\varpi_{0i}
\exp\!\left[
-\lambda_{\varpi i}
\bigl((t+\tau_{\varpi i})^{a_{\varpi i}}-\tau_{\varpi i}^{a_{\varpi i}}\bigr)
\right],
\]
with
\[
\varpi_{0i}=0.5,\qquad
\lambda_{\varpi i}=0.15,\qquad
a_{\varpi i}=0.6,\qquad
\tau_{\varpi i}=1000.
\]
This profile is positive, continuous, decays to zero, and belongs to \(L_1[0,\infty)\), as required in Theorem~\ref{thm:network}.

For comparison on agents \(1\) and \(6\), we also implement a baseline controller without any dedicated resiliency mechanism. The non-resilient comparison controller uses the same command filter, immersion matrices, and nominal embedding terms as the resilient controller, but removes the adaptive virtual-actuator recovery law. Specifically, after defining \(e_{x_i}=x_i-\Pi_{i_1}z_i\) and \(s_i=B_i^TP_ie_{x_i}\), the recovery input is replaced by the fixed feedback
\[
u_{i,\mathrm{nr}}^r=-k_{\mathrm{nom},i}s_i,
\qquad
k_{\mathrm{nom},i}=\frac{\alpha_i}{2}.
\]
Therefore, the non-resilient nominal actuator command is
\[
u_{i,\mathrm{nom}}^{\mathrm{nr}}
=
\psi_i^{-1}
\left(
Q_iz_i+K_ir_i-k_{\mathrm{nom},i}B_i^TP_ie_{x_i}
\right).
\]

\subsection{Leaders and Network Topology}
\label{subsec:sim_leaders_topology}
The three leaders' trajectories are generated by bounded smooth geometric paths \(\bar\phi_\ell(\tau)\in\mathbb R^2\), \(\ell=1,2,3\), with
\[
\phi_\ell(t)=\bar\phi_\ell(\tau(t)),
\qquad
\dot\tau(t)=0.30,
\qquad
\tau(0)=0.
\]

The leaders' paths have spatial scale \(500\,\mathrm{m}\) and are defined by
\[
\bar\phi_\ell(\tau)
=
500\left(c_0(\tau)+R(\theta(\tau))q_\ell(\tau)\right),
\qquad
\ell=1,2,3 .
\]
The center motion is
\[
\begin{split}
&c_0(\tau)=\\
&\begin{bmatrix}
1.10\sin(0.050\tau)+0.18\sin(0.130\tau)\\
0.75\sin(0.085\tau+0.25\sin(0.030\tau))+0.14\cos(0.120\tau)
\end{bmatrix},
\end{split}
\]
and the rotation angle is
\[
\theta(\tau)=0.065\tau+0.18\sin(0.035\tau).
\]
The relative leader offsets are
\[
q_\ell(\tau)
=
c(\tau)S(\tau)r_\ell(\tau)
\begin{bmatrix}
\cos a_\ell(\tau)\\
\sin a_\ell(\tau)
\end{bmatrix},
\]
with resizing factor
\[
c(\tau)=0.18+0.82
\left(
\frac{1+\cos(0.055\tau)}{2}
\right)^2,
\]
the shape matrix
\[
S(\tau)=
\begin{bmatrix}
s_{11}(\tau)&s_{12}(\tau)\\
0&s_{22}(\tau)
\end{bmatrix},
\]
and
\[
\begin{aligned}
s_{11}(\tau)&=1.00+0.18\sin(0.038\tau),\\
s_{12}(\tau)&=0.12\sin(0.050\tau+0.30),\\
s_{22}(\tau)&=0.82+0.16\cos(0.044\tau+0.50).
\end{aligned}
\]
The angular modulations are
\[
\begin{aligned}
a_1(\tau)&=0.20\sin(0.041\tau),\\
a_2(\tau)&=\tfrac{2\pi}{3}+0.28\sin(0.047\tau+0.90),\\
a_3(\tau)&=\tfrac{4\pi}{3}+0.24\cos(0.052\tau+0.40),
\end{aligned}
\]
and the radial modulations are
\[
\begin{aligned}
r_1(\tau)&=0.72+0.18\sin(0.060\tau+0.20),\\
r_2(\tau)&=0.62+0.16\cos(0.073\tau+1.10),\\
r_3(\tau)&=0.68+0.20\sin(0.067\tau+2.00).
\end{aligned}
\]
All path-defining functions and their \(\tau\)-derivatives are bounded. Since \(\dot\tau=0.30\), the generated signals \(\phi_\ell(t)\) are bounded and locally absolutely continuous with bounded derivatives, and therefore satisfy Assumption~\ref{Assumption: Leaders bound}.

The network topology is depicted in Figure~\ref{fig:sim_network_topology}. The directed interaction graph uses leader-to-follower weights
\[
a_{1,1}^L=0.40,\qquad a_{2,2}^L=0.20,\qquad a_{4,3}^L=0.30,
\]
and follower-to-follower weights
\[
\begin{gathered}
a_{2,1}^F=0.35,\quad a_{3,1}^F=0.25,\quad a_{3,2}^F=0.20,\quad a_{4,2}^F=0.30,\\
a_{5,3}^F=0.25,\quad a_{5,4}^F=0.25,\quad a_{6,5}^F=0.35,\quad a_{6,2}^F=0.15 .
\end{gathered}
\]
The graph has a leader-rooted united spanning tree, so Assumption~\ref{Assumption: Graph} holds. The induced matrix \(W_L=-H_F^{-1}L_{FL}\) is nonnegative and row stochastic to numerical precision, with maximum row-sum error \(1.1\times10^{-16}\). Hence the graph-induced target
\[
\sigma_i^\star(t)=\sum_{\ell=1}^{3}W_{i\ell}\phi_\ell(t)
\]
lies in \(\co(\Phi_L(t))\) for every follower and every instant. The gauge vector \(p=H_F^{-\top}\mathbf 1_M\) used in Theorem~\ref{thm:network} is positive entrywise, with \(\min_i p_i=0.50\) and \(\max_i p_i=2.4793\).

\begin{figure}[t]
\centering
\begin{tikzpicture}[scale=0.70, transform shape, >=Stealth, every node/.style={circle, draw, minimum size=8mm}]
\tikzset{
midarrow/.style={
decoration={markings, mark=at position 0.5 with {\arrow[scale=1.7]{Stealth}}},
postaction={decorate}
}
}
\node (L2) at (0,0) {L2};
\node (L1) at (-2,0) {L1};
\node (L3) at (2,0) {L3};
\node (A4) at (0,-2) {F4};
\node (A2) at (-2,-2) {F2};
\node (A1) at (-4,-2) {F1};
\node (A6) at (2,-4) {F6};
\node (A3) at (-2,-4) {F3};
\node (A5) at (0,-4) {F5};
\draw[midarrow] (L2) -- (A2);
\draw[midarrow] (L1) -- (A1);
\draw[midarrow] (L3) -- (A4);
\draw[midarrow] (A1) -- (A2);
\draw[midarrow] (A1) -- (A3);
\draw[midarrow] (A2) -- (A3);
\draw[midarrow] (A2) -- (A4);
\draw[midarrow] (A2) -- (A6);
\draw[midarrow] (A4) -- (A5);
\draw[midarrow] (A3) -- (A5);
\draw[midarrow] (A5) -- (A6);
\end{tikzpicture}
\caption{Directed communication graph linking six followers and three leaders.}
\label{fig:sim_network_topology}
\end{figure}
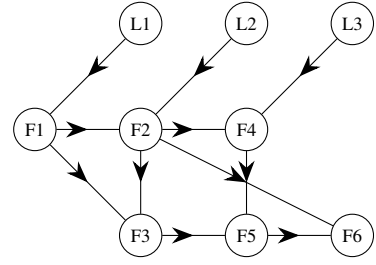

\subsection{Results and Discussion}
\label{subsec:sim_results}
We run the simulation for \(2000\,\mathrm{s}\), choosing the verification interval
\[
\mathcal T_v=[1800,2000]\,\mathrm{s}
\]
to report the quantitative maxima below. As depicted in Figure~\ref{fig:Leaders}, over the \(2000\,\mathrm{s}\) run, the maximum leader speed is \(24.0\,\mathrm{m/s}\), the maximum leader-position norm is \(1.11\times10^3\,\mathrm{m}\), the hull area ranges over \([2.72\times10^3,1.83\times10^5]\,\mathrm{m}^2\), the minimum pairwise leader distance is \(55.9\,\mathrm{m}\), and the maximum hull diameter is \(706\,\mathrm{m}\).

\begin{figure}[t]
\centering
\includegraphics[width=\linewidth]{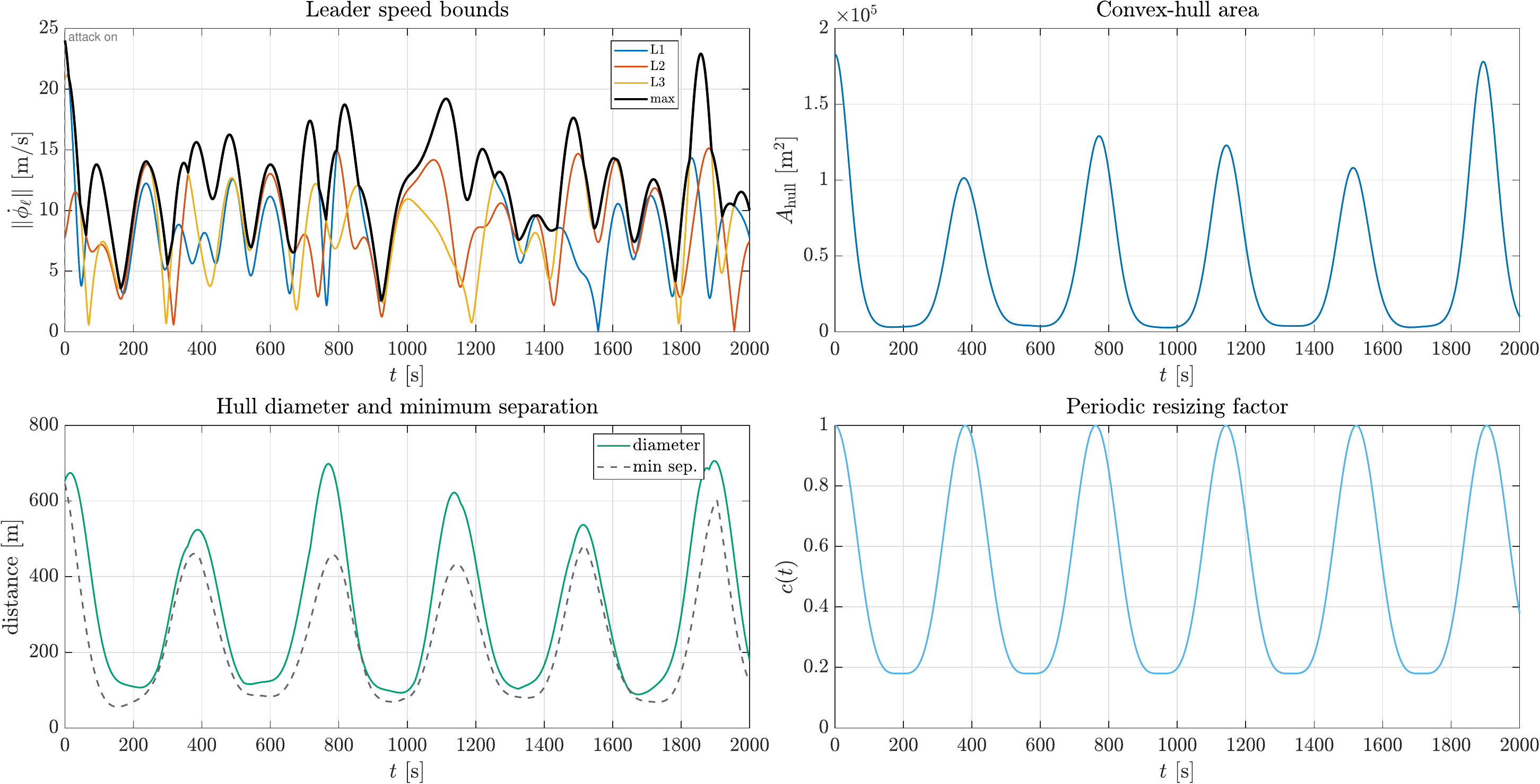}
\caption{Leaders' trajectory characteristics over the \(2000\,\mathrm{s}\) run: leaders' speed (top left), convex-hull area (top right), hull diameter (bottom left-solid), minimum pairwise separation of the leaders' positions (bottom left-dashed), and the leaders' resizing factor \(c(t)\) (bottom right).}
\label{fig:Leaders}
\end{figure}

Figures~\ref{fig:sim_trajectory_y} and~\ref{fig:sim_trajectory_x} show the follower \(y\)- and \(x\)-coordinate trajectories, respectively. The plots distinguish the command states \(\sigma_i\), the admissible-filter outputs \(H_iz_i\), and the physical outputs \(y_i\). The command states approach the leader convex hull asymptotically, whereas the physical outputs inherit the practical residual characterized in Corollary~\ref{cor:main}.

\begin{figure*}[t]
\centering
\includegraphics[width=0.9\textwidth]{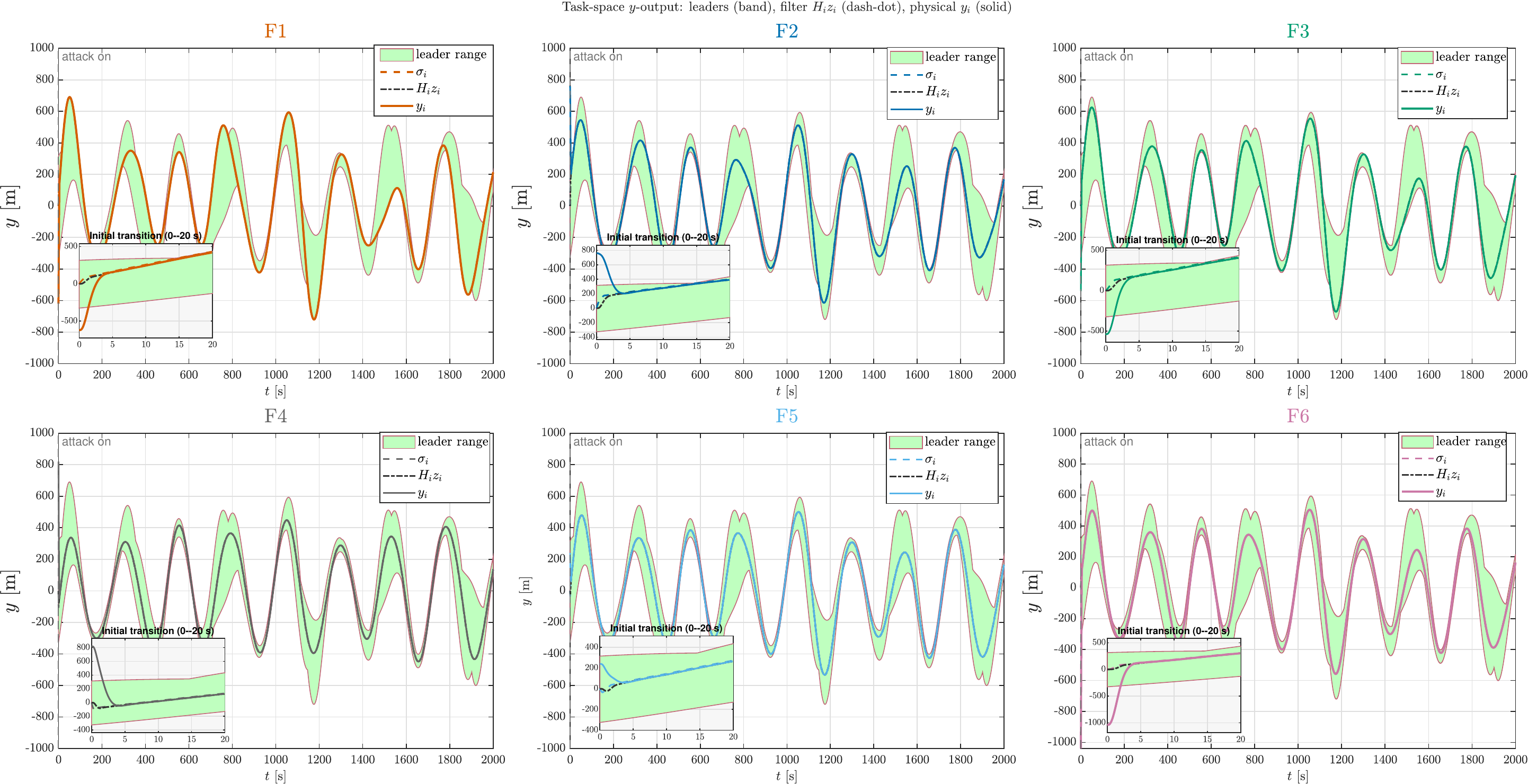}
\caption{Time-domain \(y\)-coordinate trajectories. The shaded band shows the instantaneous leader-coordinate range projected in the corresponding coordinate. The command states $\sigma_i$, filter outputs $H_iz_i$, and physical outputs $y_i$ are plotted to distinguish command containment from physical-output realization.}
\label{fig:sim_trajectory_y}
\end{figure*}

\begin{figure*}[t]
\centering
\includegraphics[width=0.9\textwidth]{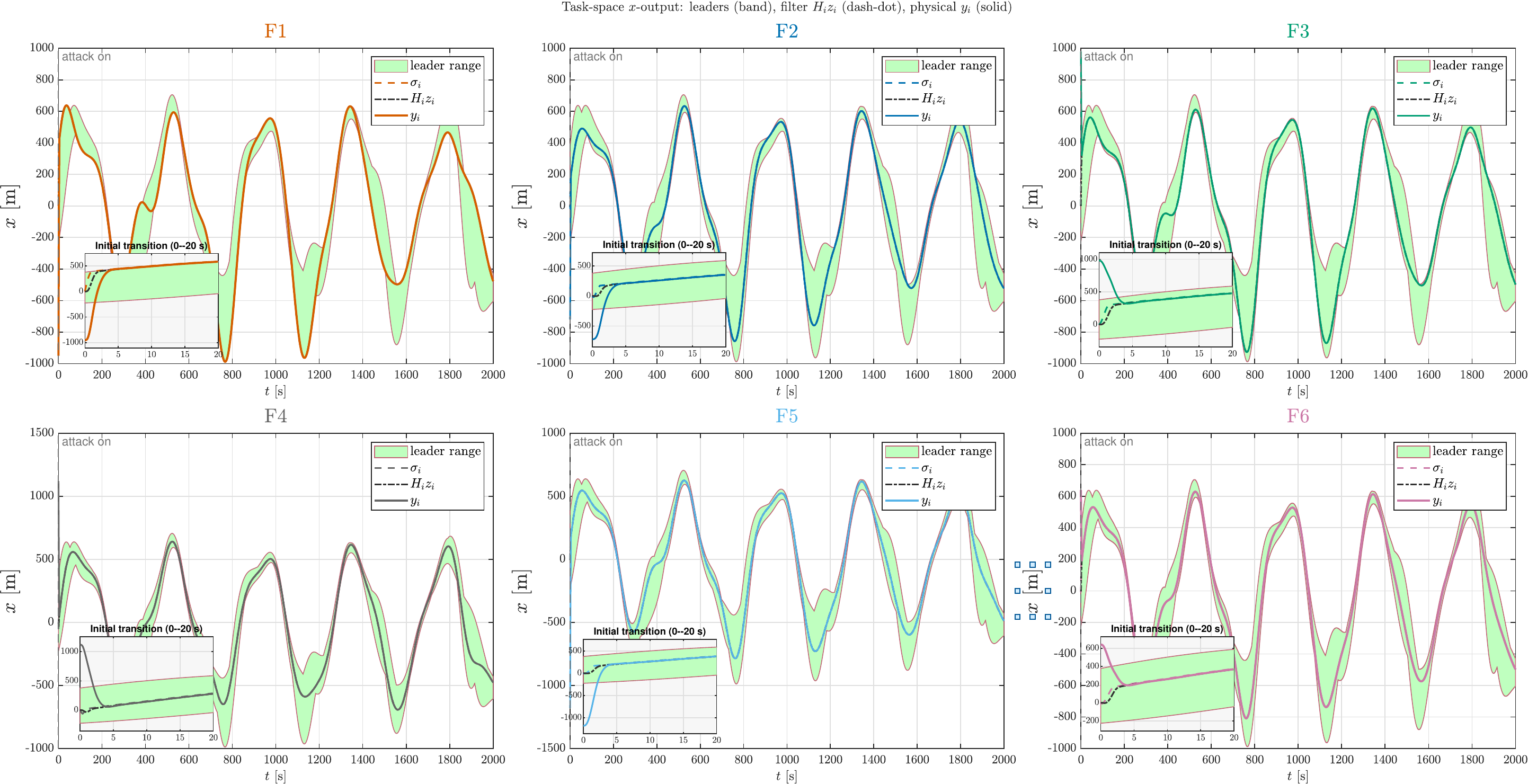}
\caption{Time-domain \(x\)-coordinate trajectories. The shaded band shows the instantaneous leader-coordinate range projected in the corresponding coordinate. The command states $\sigma_i$, filter outputs $H_iz_i$, and physical outputs $y_i$ are plotted to distinguish command containment from physical-output realization.}
\label{fig:sim_trajectory_x}
\end{figure*}

To evaluate the network-interface performance, we use
\[
E_{\sigma_i}(t)=\|\sigma_i(t)-\sigma_i^\star(t)\|,\qquad
E_\sigma(t)=\max_i E_{\sigma_i}(t),
\]
and
\[
D_{\sigma_i}(t)=
\dist\bigl(\sigma_i(t),\co(\Phi_L(t))\bigr),\qquad
D_\sigma(t)=\max_iD_{\sigma_i}(t).
\]
The corresponding values are depicted in Figure~\ref{fig:sim_command}. Since \(\sigma_i^\star(t)\in\co(\Phi_L(t))\), one has
\[
D_\sigma(t)\le E_\sigma(t).
\]
On \(\mathcal T_v\),
\[
\max_{t\in\mathcal T_v}D_\sigma(t)=2.33\times10^{-3}\,\mathrm{m},
\]
which is \(3.30\times10^{-6}\) of the maximum leader-hull diameter. At the final time,
\[
E_\sigma(2000)=1.50\times10^{-5}\,\mathrm{m},
\qquad
D_\sigma(2000)=4.26\times10^{-7}\,\mathrm{m}.
\]
These values support the command-containment conclusion of Theorem~\ref{thm:network} over the finite simulation horizon. Considering the zero initial conditions for \(\sigma_i\) and noting that the initial leaders' convex hull includes the origin, initially \(D_{\sigma_i}(0)=0\) for all the agents; however, since \(\sigma_i(0)\) is different from the dedicated point \(\sigma_i^\star(0)\), \(E_{\sigma_i}(0)\) is non-zero. As the agents start tracking \(\sigma_i^\star(t)\) and the leaders start moving in the task space, \(D_{\sigma_i}\) first increases for the agents, and then, as \(\sigma_i\) approaches \(\sigma_i^\star\) for all the agents, \(E_\sigma\) and \(D_\sigma\) both converge to zero. Moreover, since the adaptive protocol is not tuned using an a priori bound on the leaders' velocities, its gains increase only when the realized motion requires stronger interaction. Consequently, when the leaders reach velocity levels or velocity variations not encountered earlier in the finite simulation horizon, a small adaptation transient can appear in \(E_\sigma\) and \(D_\sigma\). These variations are consistent with the protocol's online nature and do not indicate a loss of command containment. Specifically, these transitions are expected to correlate with the leaders' speed, as illustrated in Figure~\ref{fig:Leaders}.

\begin{figure}[t]
\centering
\includegraphics[width=\linewidth]{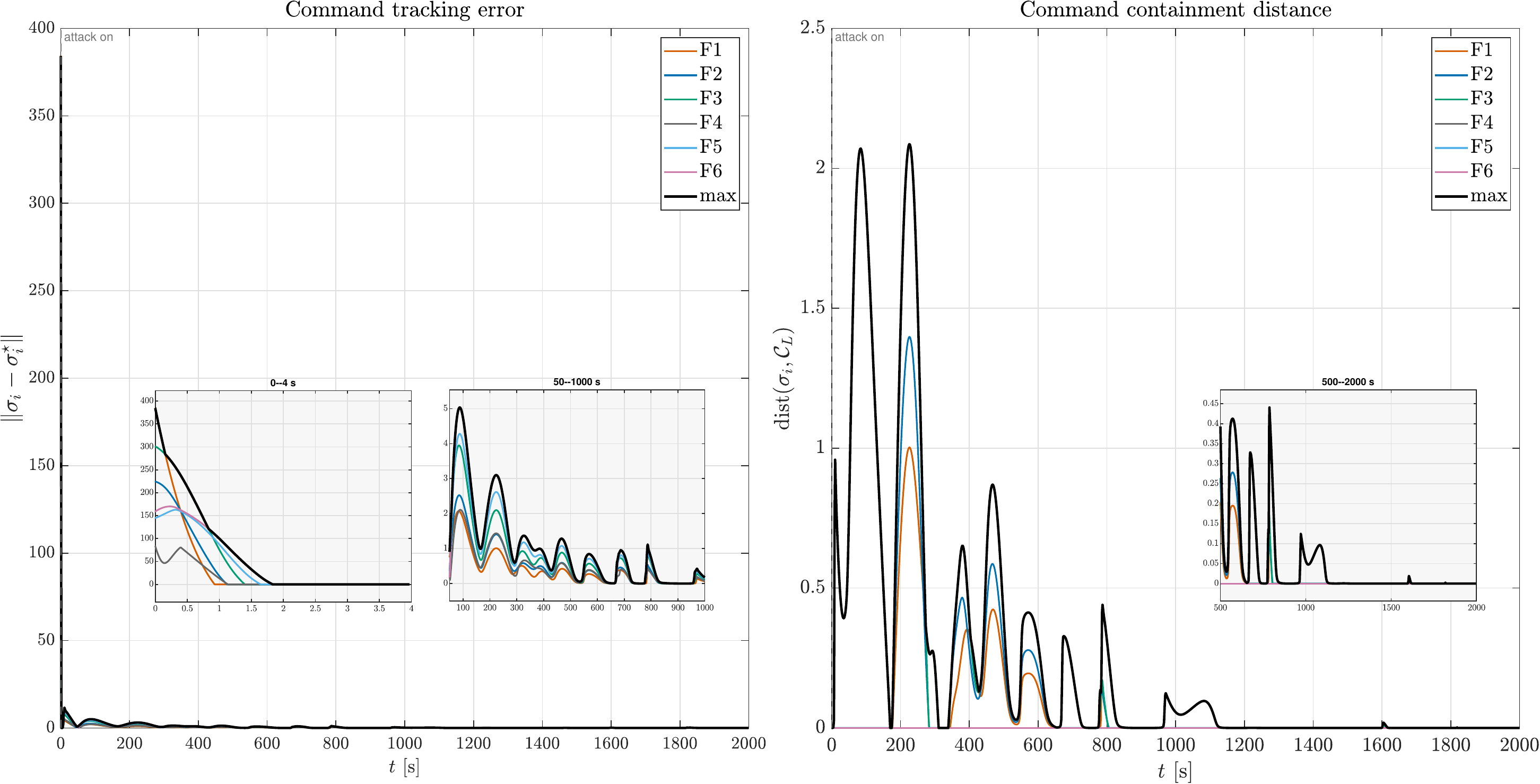}
\caption{The curve \(E_\sigma\) (left) measures convergence to the graph-induced convex-combination targets \(\sigma_i^\star\), while \(D_\sigma\) (right) measures the point-to-set distance from the command states to the leader convex hull.}
\label{fig:sim_command}
\end{figure}

Figure~\ref{fig:sim_network} reports the local disagreement variables \(\vartheta_i\) and the command rates \(\dot\sigma_i\). These quantities have different roles. The variables \(\vartheta_i\) are the local containment disagreements suppressed by the distributed protocol, whereas \(\dot\sigma_i\) measures the task-space motion of the generated commands and need not converge to zero when the leader hull continues to move. At \(t=2000\,\mathrm{s}\), the maximum command rate is \(8.60\,\mathrm{m/s}\), while the maximum leader speed over the simulation is \(24.0\,\mathrm{m/s}\). On \(\mathcal T_v\), the maximum disagreement norm is \(3.38\times10^{-1}\), and the maximum command rate is \(16.3\,\mathrm{m/s}\).

\begin{figure}[t]
\centering
\includegraphics[width=\linewidth]{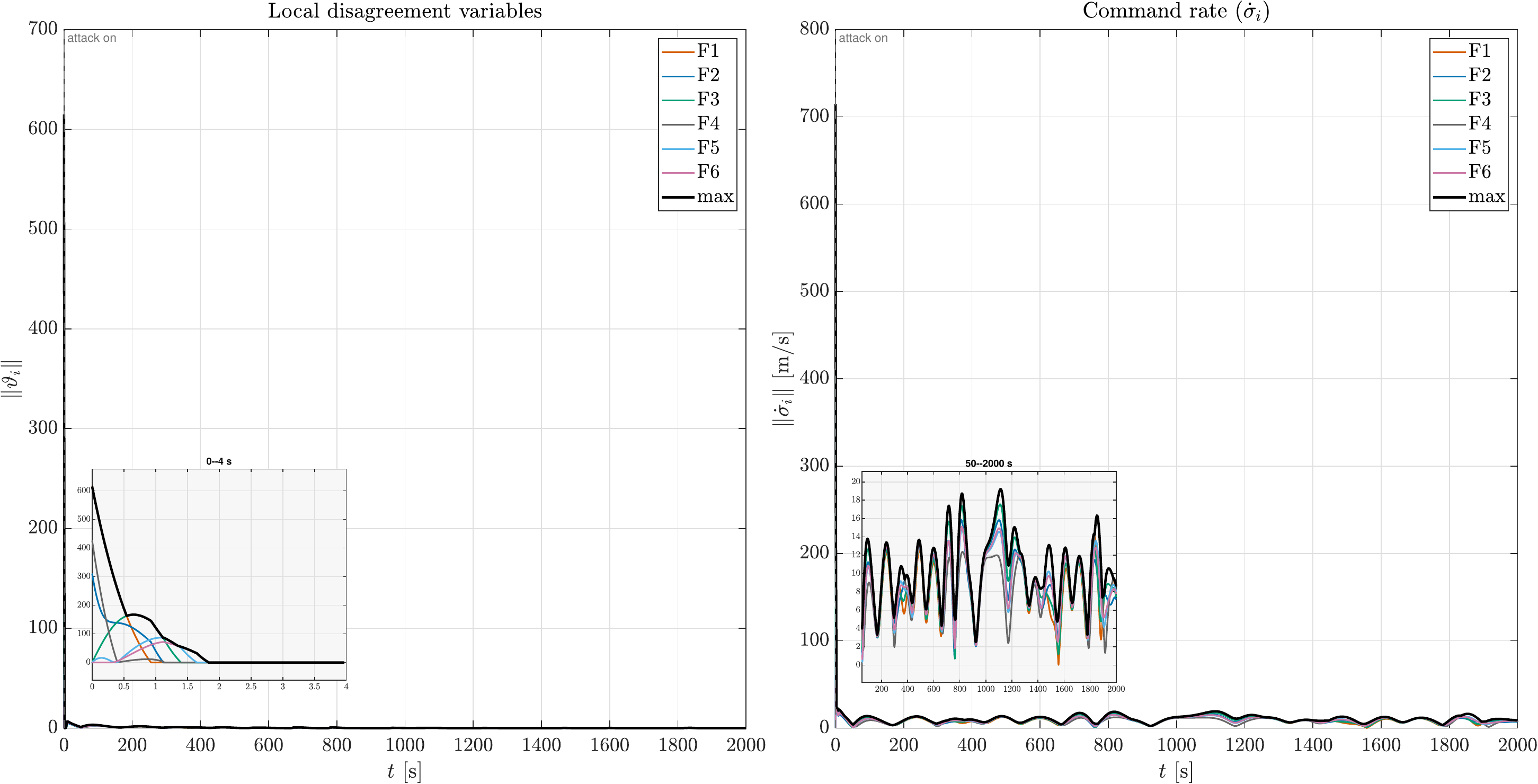}
\caption{The disagreement variables \(\vartheta_i\) (left) are reduced by the distributed protocol, while the command rates \(\dot\sigma_i\) (right) remain time varying because the leader-generated hull continues to move.}
\label{fig:sim_network}
\end{figure}

Figure~\ref{fig:sim_physical} reports the local-level performance through \(D_{y_i}(t)=\dist(y_i(t),\co(\Phi_L(t)))\), \(\|y_i-\sigma_i^\star\|\), \(\|y_i-\sigma_i\|\), and \(\|H_iz_i-\sigma_i\|\). The first two metrics describe the overall containment performance supported by Corollary~\ref{cor:main}. The last two quantify local realization errors: \(\|y_i-\sigma_i\|\) measures the output-command mismatch, while \(\|H_iz_i-\sigma_i\|\) is the command-filter realization residual.

\begin{figure*}[t]
\centering
\includegraphics[width=\textwidth]{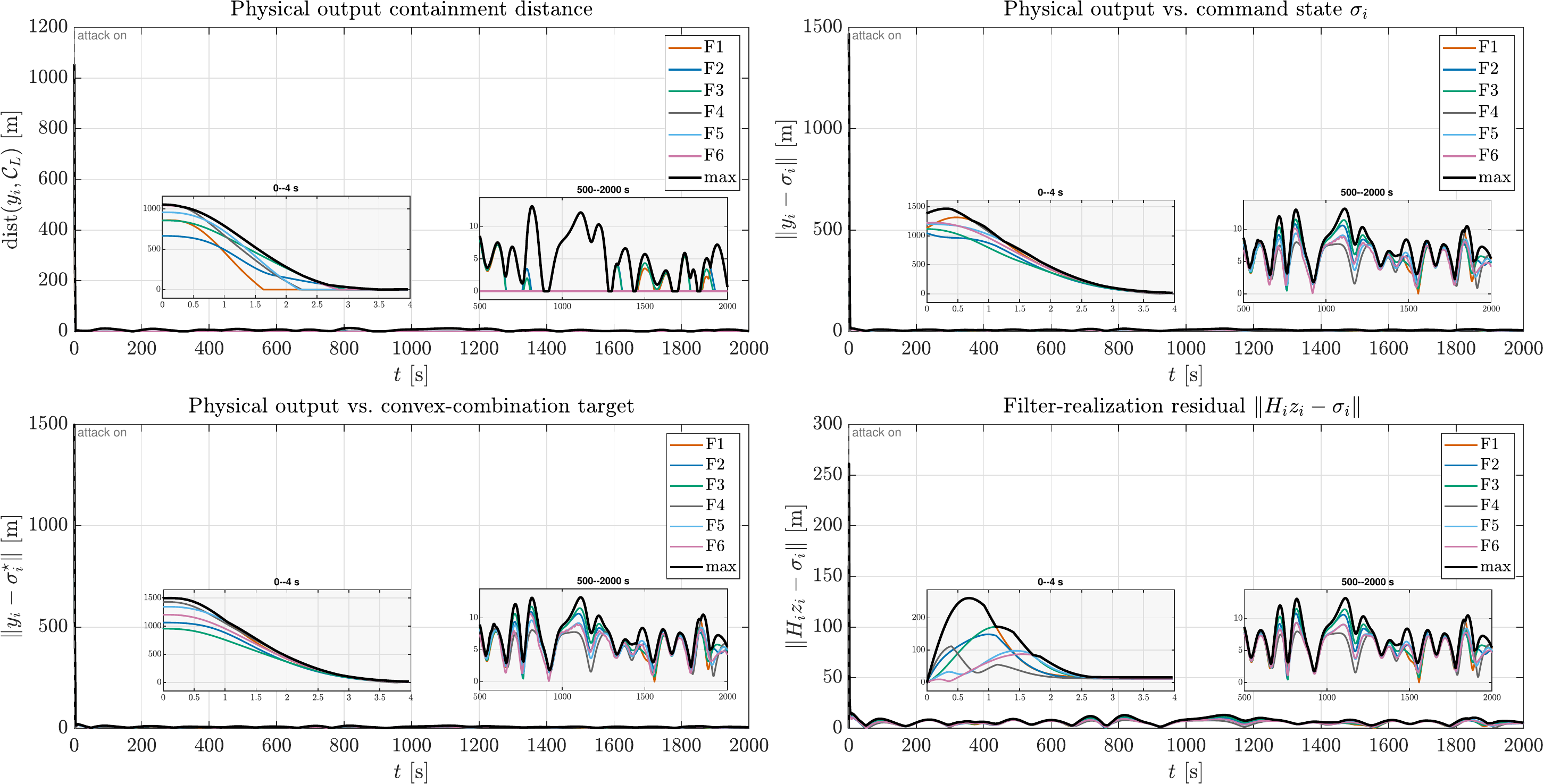}
\caption{The figure shows the distance of \(y_i\) to the leader hull, the output-command error, the error relative to the graph-induced target, and the command-filter realization residual.}
\label{fig:sim_physical}
\end{figure*}

Figure~\ref{fig:sim_decomposition} shows the residual decomposition associated with Corollary~\ref{cor:main}. We use
\[
D_y(t)=\max_i D_{y_i}(t)
\]
as the overall containment-performance measure. Over the interval \(\mathcal T_v\), we have
\[
\max_{t\in\mathcal T_v}D_y(t)=7.23\,\mathrm{m},
\]
which is approximately \(1.02\%\) of the maximum leader-hull diameter \(706\,\mathrm{m}\). On the same interval,
\[
\begin{split}
&\max_{i,\ t\in\mathcal T_v}\|H_i z_i(t)-\sigma_i(t)\|=10.5\,\mathrm{m},\\
&\max_{i,\ t\in\mathcal T_v}\|y_i(t)-H_i z_i(t)\|=0.943\,\mathrm{m}.
\end{split}
\]
Therefore, the dominant term in the containment residual is the command-filter realization residual \(H_i z_i-\sigma_i\), not the command-layer containment error. The plotted bound is
\[
\begin{aligned}
D_y(t)
\le\;&
D_\sigma(t)
+\max_i\|H_i z_i(t)-\sigma_i(t)\|\\
&+\max_i\|y_i(t)-H_i z_i(t)\|.
\end{aligned}
\]
On \(\mathcal T_v\), with \(\max_{t\in\mathcal T_v}D_\sigma(t)=2.33\times10^{-3}\,\mathrm{m}\), this agrees with Corollary~\ref{cor:main}, where the network layer produces near-exact command containment, while the physical outputs retain a practical residual due to the local realization of moving commands by heterogeneous relative-degree-four followers. The estimate in Corollary~\ref{cor:main} holds per follower: with the kernel gain \(\kappa_i=\int_0^\infty\|H_ie^{F_is}X_i\|\,ds\in[0.62,0.70]\), the predicted residual \(\varepsilon_i=\kappa_i\sup_{\mathcal T_v}\|\dot\sigma_i\|\) bounds the realized command-filter residual \(\max_{\mathcal T_v}\|H_iz_i-\sigma_i\|\) for every \(i\); for instance \((\text{realized},\varepsilon_i)=(9.82,10.04)\,\mathrm{m}\) for F1 and \((7.82,7.93)\,\mathrm{m}\) for F6.

\begin{figure}[t]
\centering
\includegraphics[width=\linewidth]{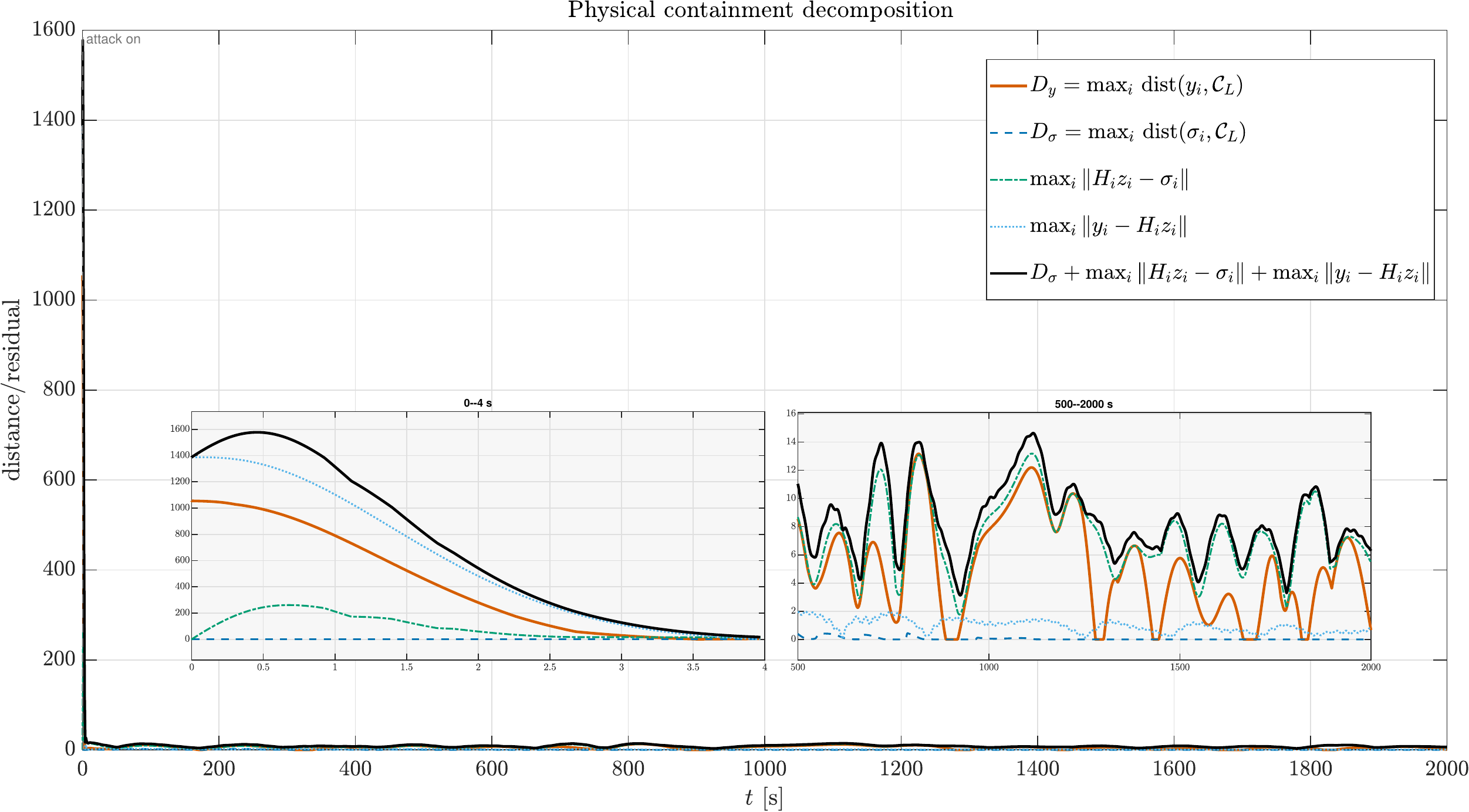}
\caption{Residual decomposition of Corollary~\ref{cor:main}. The physical containment distance is shown together with the command-containment term, the command-filter realization residual, the local tracking residual, and the decomposition upper bound.}
\label{fig:sim_decomposition}
\end{figure}

Figure~\ref{fig:sim_local_recovery} reports the error norms \(\|e_{x_i}\|\) and \(\|e_{\eta_i}\|\), which provide conservative indicators of the local recovery performance under actuator attacks. On the interval \(\mathcal T_v\), we have
\[
\max_{i,\ t\in\mathcal T_v}\|e_{x_i}(t)\|=0.943,
\qquad
\max_{i,\ t\in\mathcal T_v}\|e_{\eta_i}(t)\|=7.73 .
\]
Both maxima occur for follower~6, which is subjected to the most severe local attack. Since \(e_{x_i}\) contains the output coordinates and their derivatives up to order three, \(\|e_{x_i}\|\) is a normal-form tracking-error norm, not a pure position error. The larger value of \(\|e_{\eta_6}\|\) reflects the response of the stable zero-dynamics coordinates.

\begin{figure}[t]
\centering
\includegraphics[width=\linewidth]{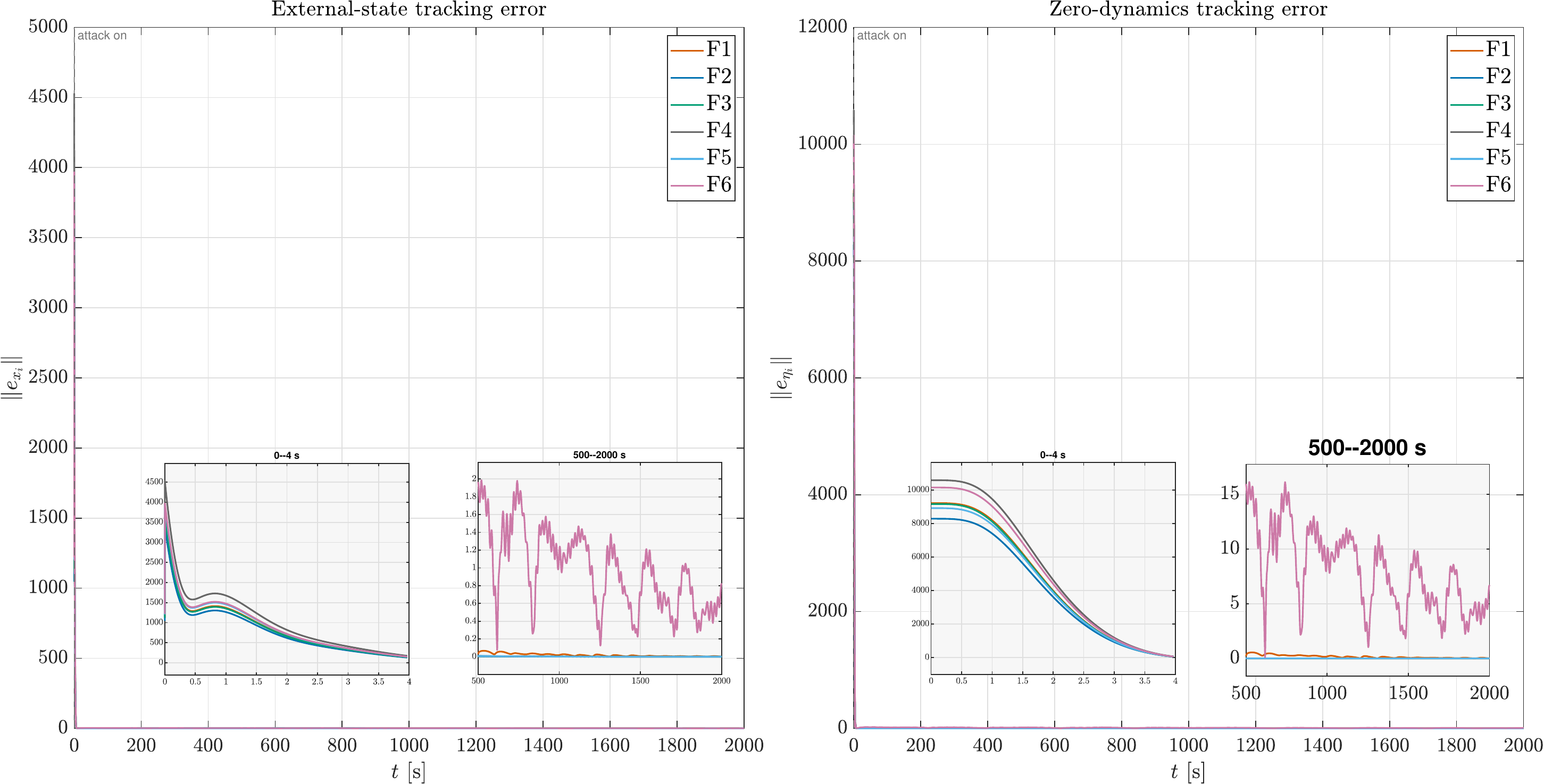}
\caption{Follower-wise local recovery errors. The two panels show \(\|e_{x_i}\|\) and \(\|e_{\eta_i}\|\) for each follower, identifying follower~6 as the source of the largest local residuals.}
\label{fig:sim_local_recovery}
\end{figure}

Finally, Figures~\ref{fig:sim_nonres_y} and~\ref{fig:sim_nonres_x} compare the resilient and non-resilient responses of the two attacked followers, agents \(1\) and \(6\). The figures also show the separation between the trajectories generated with the adaptive virtual actuator and those generated by the non-resilient baseline. Fo follower~1 the maximum resilient versus non-resilient baseline separation is \(2.50\,\mathrm{m}\), and for follower~6 the maximum resilient versus non-resilient separation is \(7.57{\times}10^3\,\mathrm{m}\).

\begin{figure*}[t]
\centering
\includegraphics[width=\textwidth]{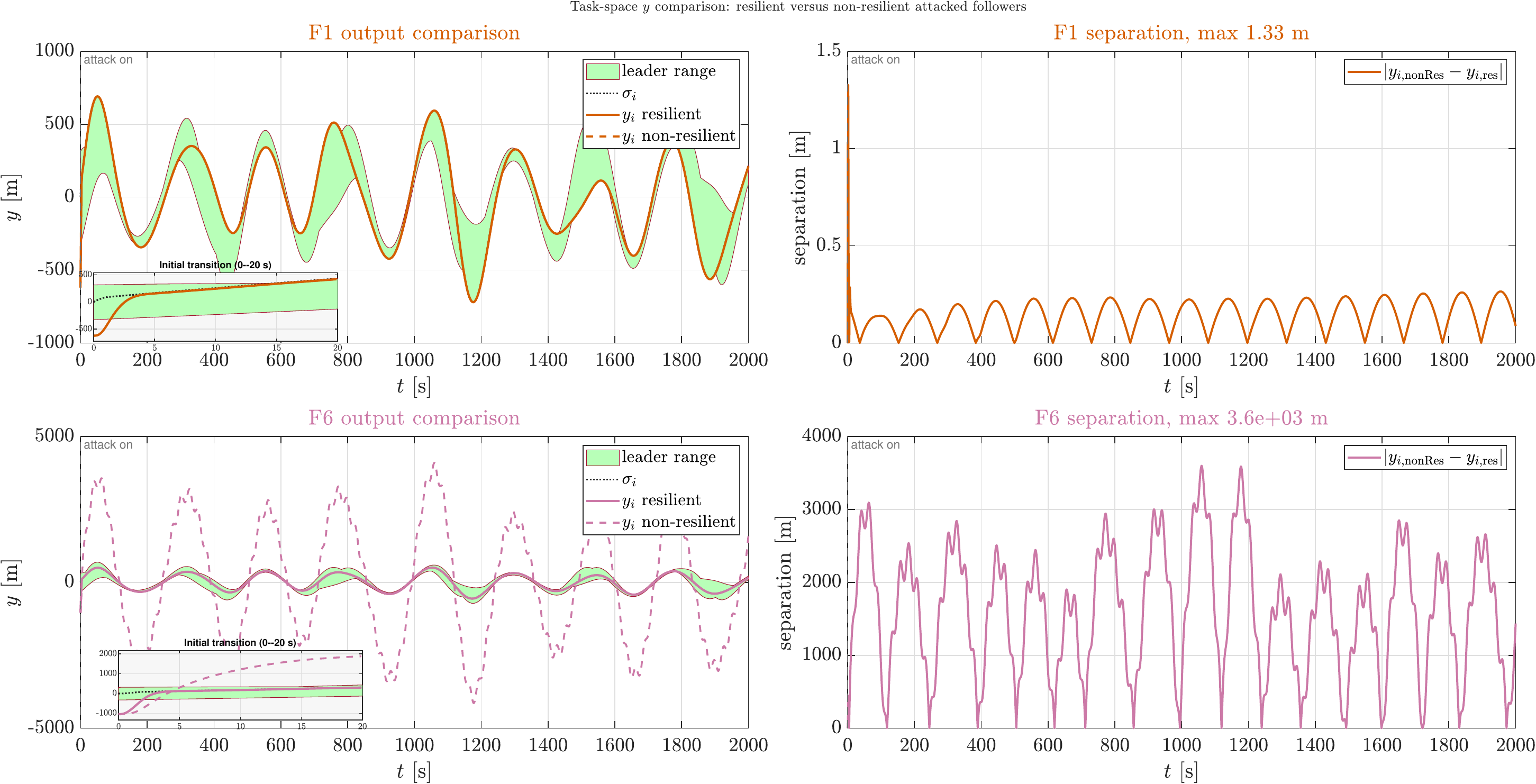}
\caption{Task-space \(y\)-coordinate comparison for the attacked followers. The solid curves show resilient physical outputs, the dashed curves show non-resilient physical outputs, the dotted curves show generated commands \(\sigma_i\), and the green band indicates the instantaneous leader-coordinate range. The right panels show the resilient/non-resilient separation.}
\label{fig:sim_nonres_y}
\end{figure*}

\begin{figure*}[t]
\centering
\includegraphics[width=\textwidth]{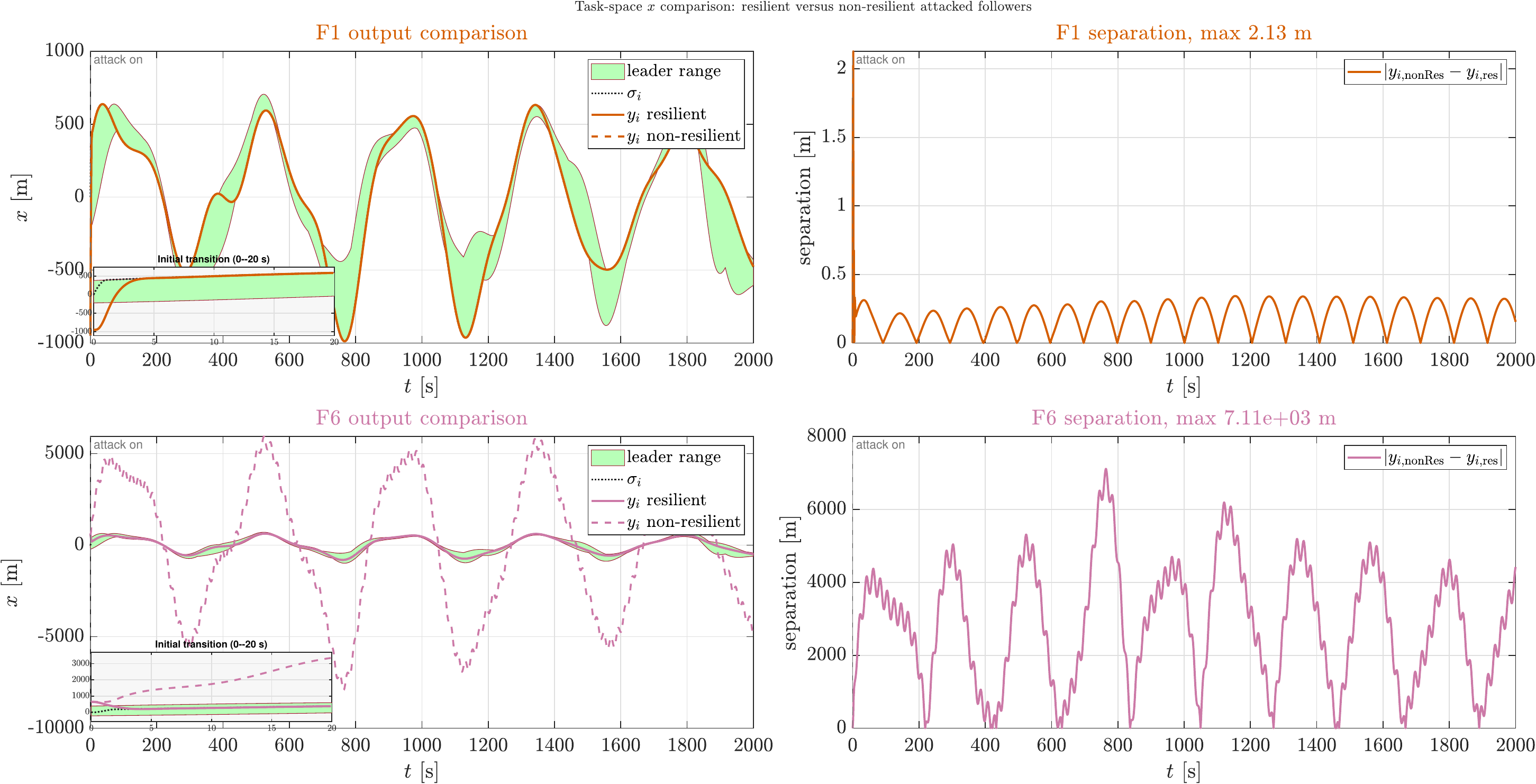}
\caption{Task-space \(x\)-coordinate comparison for the attacked followers. The solid curves show resilient physical outputs, the dashed curves show non-resilient physical outputs, the dotted curves show generated commands \(\sigma_i\), and the green band indicates the instantaneous leader-coordinate range. The right panels show the resilient/non-resilient separation.}
\label{fig:sim_nonres_x}
\end{figure*}

The per-follower metric values over \(\mathcal T_v\) in Table~\ref{tab:sim_per_follower} clarify the roles of the command layer, the command filter, and the local recovery layer. Define
\[
\begin{split}
&d_i^\Phi
=
\max_{t\in\mathcal T_v}
\dist\bigl(y_i(t),\co(\Phi_L(t))\bigr),\ e_{\eta_i}^{\max}
=
\max_{t\in\mathcal T_v}
\|e_{\eta_i}(t)\|,\\
&e_{x_i}^{\max}
=
\max_{t\in\mathcal T_v}
\|e_{x_i}(t)\|,\ \delta_i^{\mathrm{loc}}
=
\max_{t\in\mathcal T_v}
\|y_i(t)-H_i z_i(t)\|.
\end{split}
\]
The largest point-to-set distance \(d_i^\Phi\) occurs for follower~1, whereas the largest local tracking errors occur for follower~6. The quantity \(d_i^\Phi\) depends on where the realized physical output lies relative to the moving leader hull, whereas the local errors measure how well the attacked plant follows its own command and filtered command. For the two attacked followers the external error is concentrated almost entirely in the output coordinates, i.e.\ \(\|C_ie_{x_i}\|\approx\|e_{x_i}\|\) on \(\mathcal T_v\); hence \(\delta_i^{\mathrm{loc}}\) and \(e_{x_i}^{\max}\) coincide to the reported precision for F1 and F6, whereas for the unattacked followers the output component accounts for only about \(8\%\) of \(\|e_{x_i}\|\).

\begin{table}[t]
\centering
\caption{Per-follower maxima on \(\mathcal T_v=[1800,2000]\,\mathrm{s}\).}
\label{tab:sim_per_follower}
\scriptsize
\setlength{\tabcolsep}{3.0pt}
\begin{tabular}{@{}lrrrr@{}}
\hline
Follower
& \(d_i^\Phi\)
& \(\delta_i^{\mathrm{loc}}\)
& \(e_{x_i}^{\max}\)
& \(e_{\eta_i}^{\max}\) \\
\hline
F1 & \(7.23\) & \(7.55{\times}10^{-3}\) & \(7.55{\times}10^{-3}\) & \(6.34{\times}10^{-2}\) \\
F2 & \(5.10\) & \(1.34{\times}10^{-5}\) & \(1.67{\times}10^{-4}\) & \(1.14{\times}10^{-4}\) \\
F3 & \(3.41\) & \(1.14{\times}10^{-5}\) & \(1.43{\times}10^{-4}\) & \(9.39{\times}10^{-5}\) \\
F4 & \(0.00\) & \(1.13{\times}10^{-5}\) & \(1.46{\times}10^{-4}\) & \(9.48{\times}10^{-5}\) \\
F5 & \(0.00\) & \(1.43{\times}10^{-5}\) & \(1.80{\times}10^{-4}\) & \(1.16{\times}10^{-4}\) \\
F6 & \(0.00\) & \(9.43{\times}10^{-1}\) & \(9.43{\times}10^{-1}\) & \(7.73\) \\
\hline
\end{tabular}
\end{table}

Table~\ref{tab:sim_metrics} summarizes the main evaluation metrics. The command layer reaches a final point-to-set distance of \(4.26\times10^{-7}\,\mathrm{m}\), while the physical outputs achieve practical containment with maximum verification-interval distance \(7.23\,\mathrm{m}\). The decomposition shows that the dominant term in the residual bound is the command-filter residual \(H_i z_i-\sigma_i\).

\begin{table}[t]
\centering
\caption{Main evaluation metric values.}
\label{tab:sim_metrics}
\scriptsize
\setlength{\tabcolsep}{3.0pt}
\begin{tabular}{@{}lr@{}}
\hline
Quantity & Value \\
\hline
Simulation horizon & \(2000\,\mathrm{s}\) \\
Verification interval & \([1800,2000]\,\mathrm{s}\) \\
Maximum leader speed & \(24.0\,\mathrm{m/s}\) \\
Maximum leader-hull diameter & \(706\,\mathrm{m}\) \\
Final \(E_\sigma\) & \(1.50{\times}10^{-5}\,\mathrm{m}\) \\
Final \(D_\sigma\) & \(4.26{\times}10^{-7}\,\mathrm{m}\) \\
Max. \(D_\sigma\) on \(\mathcal T_v\) & \(2.33{\times}10^{-3}\,\mathrm{m}\) \\
Max. \(D_y\) on \(\mathcal T_v\) & \(7.23\,\mathrm{m}\) \\
Max. \(\|H_i z_i-\sigma_i\|\) on \(\mathcal T_v\) & \(10.5\,\mathrm{m}\) \\
Max. \(\|y_i-H_i z_i\|\) on \(\mathcal T_v\) & \(0.943\,\mathrm{m}\) \\
\hline
\end{tabular}
\end{table}

\section{Conclusion}
This work developed a continuous two-layer architecture for resilient practical output containment of heterogeneous linear MIMO followers under actuator attacks and leader-model nondisclosure. The local virtual-actuator layer compensates for state-correlated, input-correlated, and bounded exogenous actuator attack effects in the command-filter tracking-error dynamics. The network layer achieves asymptotic containment of the generated commands over directed leader-rooted graphs without leader-dynamics reconstruction, known leader-motion bounds, or global graph knowledge. The physical-output statement remains practical because implementing the generated task-space command through the selected stable local filter induces the residual characterized in Corollary~\ref{cor:main}.
\bibliographystyle{IEEEtran}
\bibliography{Bib4}
\end{document}